\documentclass[11pt,letterpaper]{article}

\usepackage{booktabs} 
\usepackage{enumerate}
\usepackage{authblk}

\usepackage[margin=1in]{geometry}
\usepackage{amsmath,amsthm,amsfonts,amssymb,color,hyperref,epstopdf}
\usepackage[usenames,dvipsnames]{xcolor}

\newcommand{\rr}{\mathbb{R}}
\newcommand{\rrplus}{\rr^+}
\newcommand{\nn}{\mathbb{N}}

\newcommand{\ep}{\epsilon}

\newcommand{\ra}{\rightarrow}

\newcommand{\hide}[1]{}

\newcommand{\expect}[1]{\mathbb{E}\left[ #1 \right]}
\newcommand{\expectsub}[2]{\mathbb{E}_{#1}\left[ #2 \right]}

\renewcommand{\bold}[1]{\textbf{#1}}
\newcommand{\parabold}[1]{
	
\medskip
	
\noindent\bold{#1}}

\DeclareMathOperator*{\argmax}{arg\,max}

\newcommand{\difft}[1]{\frac{\mathsf{d} #1}{\mathsf{d} t}}

\newtheorem{theorem}{Theorem}
\newtheorem{lemma}[theorem]{Lemma}
\newtheorem{corollary}[theorem]{Corollary}

\newtheorem{proposition}[theorem]{Proposition}
\newtheorem{definition}{Definition}
\newtheorem{obs}[theorem]{Observation}

\newcommand{\calC}{\mathcal{C}}

\newcommand{\calG}{\mathcal{G}}
\newcommand{\calH}{\mathcal{H}}

\newcommand{\calM}{\mathcal{M}}

\newcommand{\calO}{\mathcal{O}}
\newcommand{\calP}{\mathcal{P}}

\newcommand{\calT}{\mathcal{T}}

\newcommand{\calV}{\mathcal{V}}

\newcommand{\calZ}{\mathcal{Z}}

\newcommand{\bbg}{\mathbf{g}}
\newcommand{\bbh}{\mathbf{h}}

\newcommand{\bbp}{\mathbf{p}}
\newcommand{\bbq}{\mathbf{q}}

\newcommand{\bbs}{\mathbf{s}}

\newcommand{\bbu}{\mathbf{u}}
\newcommand{\bbv}{\mathbf{v}}

\newcommand{\bbx}{\mathbf{x}}
\newcommand{\bby}{\mathbf{y}}

\newcommand{\bbA}{\mathbf{A}}
\newcommand{\bbB}{\mathbf{B}}
\newcommand{\bbC}{\mathbf{C}}

\newcommand{\bbG}{\mathbf{G}}
\newcommand{\bbH}{\mathbf{H}}
\newcommand{\bbI}{\mathbf{I}}
\newcommand{\bbJ}{\mathbf{J}}
\newcommand{\bbK}{\mathbf{K}}
\newcommand{\bbL}{\mathbf{L}}
\newcommand{\bbM}{\mathbf{M}}

\newcommand{\bbP}{\mathbf{P}}
\newcommand{\bbQ}{\mathbf{Q}}

\newcommand{\bbT}{\mathbf{T}}
\newcommand{\bbU}{\mathbf{U}}

\newcommand{\bbZ}{\mathbf{Z}}

\newcommand{\inner}[2]{\left\langle ~#1 ~,~ #2~\right\rangle}
\newcommand{\trans}{^{\mathsf{T}}}

\hide{
\newtheorem{theorem}{Theorem}[section]
\newtheorem{lemma}{Lemma}[section]
\newtheorem{definition}{Definition}[section]

\newtheorem{corollary}{Corollary}[section]
}

\begin{document}
\title{Chaos of Learning Beyond Zero-sum and Coordination\\
via Game Decompositions}
\author[1]{Yun Kuen Cheung\thanks{Yun Kuen Cheung acknowledges AcRF Tier 2 grant 2016-T2-1-170 and NRF 2018 Fellowship NRF-NRFF2018-07.}}
\author[2]{Yixin Tao \thanks{Yixin Tao acknowledges NSF grant CCF-1527568 and CCF-1909538.}}
\affil[1]{Singapore University of Technology and Design}
\affil[2]{Courant Institute, NYU}
\date{}
\maketitle

\vspace*{-0.2in}

\begin{abstract}
Machine learning processes, e.g.~``learning in games'', can be viewed as non-linear dynamical systems.
In general, such systems exhibit a wide spectrum of behaviors, ranging from stability/recurrence to the undesirable phenomena of chaos (or ``butterfly effect'').
Chaos captures sensitivity of round-off errors and can severely affect predictability and reproducibility of ML systems,
but AI/ML community's understanding of it remains rudimentary. It has a lot out there that await exploration.

Recently, Cheung and Piliouras~\cite{CP2019,CP2020} employed volume-expansion argument to show that Lyapunov chaos occurs in the cumulative payoff space,
when some popular learning algorithms, including Multiplicative Weights Update (MWU), Follow-the-Regularized-Leader (FTRL) and Optimistic MWU (OMWU),
are used in several subspaces of games, e.g.~zero-sum, coordination or graphical constant-sum games.
It is natural to ask: can these results generalize to much broader families of games?
We take on a game decomposition approach and answer the question affirmatively.

Among other results, we propose a notion of ``matrix domination'' and design a linear program,
and use them to characterize bimatrix games where MWU is Lyapunov chaotic almost everywhere.
Such family of games has positive Lebesgue measure in the bimatrix game space, indicating that chaos is a substantial issue of learning in games.
For multi-player games, we present a \emph{local equivalence of volume change} between general games and graphical games,
which is used to perform volume and chaos analyses of MWU and OMWU in potential games.
\end{abstract}

\section{Introduction}

In the developments of AI/ML, understanding how selfish agents learn in competitive game-theoretic environments is of primary interest,
and this is more strongly propelled recently due to the success of Generative Adversarial Networks (GANs).
As such, Evolutionary Game Theory (EGT)~\cite{HS1998,Sandholm2010}, a decades-old area devoted to the study of adaptive (learning) behaviors of agents
in competitive environments arising from Economics, Biology and Physics, has been brought to the attention of AI/ML community.
In contrast with the typical optimization (or no-regret) approach in AI/ML,
EGT provides us a (non-linear) dynamical-systemic perspective to understand ML processes.
This perspective is particularly helpful in studying ``learning in games'', where instability is commonly observed,
but AI/ML community currently lacks of a rigorous mean to perform the relevant analyses.

The theme of this paper is \emph{chaos}, a central notion in the study of dynamical systems that captures instability and unpredictability.
We seek broad families of games in which popular learning algorithms exhibit chaotic behaviors, by employing game decomposition techniques.
Next, we explain what are chaos and game decompositions, and why it is important for AI/ML community to understand chaos.
All missing proofs will appear in the appendix.

\parabold{Chaos.}
Chaos generally means a system becomes unpredictable in the long run; \emph{Lyapunov chaos} is one of the most popular chaos notions which captures the \emph{butterfly effect}:
when the starting point of a dynamical system is slightly perturbed, the resulting trajectories and final outcomes diverge quickly;
see Definition~\ref{def:lyapunov-chaos} for a formal definition.
Lyapunov chaos means that such system is very sensitive to round-off errors in computer simulations, and to measurement errors in real economies.
Indeed, Edward Lorenz, one of the pioneers of the modern chaos theory, started working on the topic because he found round-off errors led to devastatingly different outcomes
in his weather simulation program~\cite{Lorenz1963}.
In the context of learning-in-game, Lyapunov chaos indicates that Nash equilibrium is generally not achievable.

To see the importance for AI/ML community to understand chaos, recall that one of our primary targets is to build \emph{predictable} and \emph{reproducible} ML systems.
The soundness of a newly proposed ML system is usually supported by a theory, which is in turn supported via experimental evidences.
However, there is often a gap between theory and experiments: theory is built upon the assumption of infinite precision,
but experiments are done using \emph{finite-precision} computers where round-off occurs in every computation step.
While the round-off error per step is small, it is unclear how it can affect the final outcome of the system, let alone the cumulating effect of all the errors across multiple steps.
Moreover, how the round-off is done depends on the OS, the CPU/GPU architecture, the compiler and more, which vary from computer to computer,
so reproducing the same result across multiple computers is not as trivial as some of us had presumed.
These unpredictability and irreproducibility issues are exemplified by two quotes from Ali Rahimi's NIPS'2017 test-of-time award speech~\cite{Rahimi2017}:
\begin{quote}
\emph{``Someone on another team changed the default rounding mode of some Tensorflow internals from `truncate toward zero' to `round to even'.
Our training broke, our error rate went from less than $25\%$ error to $\sim 99.97\%$ error.''}

\emph{``If a machine learning algorithm does crazy things on a linear model, it's going to do crazy things on complex non-linear models too.''}
\end{quote}

\parabold{Game Decomposition.}
To understand how game decomposition works, we first compare the dynamical-systemic and optimization approaches. 
Historically, there is a trade-off of them between generality of settings and scope of results.\footnote{Game-theoretically,
coarse correlated equilibria (generated by time-average of no-regret algorithm)
and Nash equilibria (often fixed points of game dynamics)
might be viewed as the 
products of the two approaches.}
No-regret learning works even in general and adversarial settings, but its effectiveness is benchmarked w.r.t.~the \emph{average of history},
while it sheds little insight on the daily behaviors of the learning-in-game systems.
The dynamical-systemic approach primarily aims at understanding the daily behaviors,
but compelling results are often limited to specific families of games because non-linear dynamical systems are inherently difficult to analyze in general.

A natural approach to extend those compelling results to more general families of games is via game decomposition.
To explain how it works, suppose there is a specific family of games, denoted by $\calH$, for which some compelling results are shown.
Given a general game, we seek to decompose it into a sum of its \emph{projection} on $\calH$ and a \emph{residue} component.
If the residue is small, then it is plausible that those compelling results extend (approximately).
For instance, if the residue is small, then any Nash equilibrium of the projection is an approximate Nash equilibrium of the original game.
More generally, a game might be decomposed into three or more components which can be studied separately.
In seeking of games where learning is stable, the game decomposition approach was used in several
works~\cite{Candogan2011,Candogan2013-GEB,Candogan2013-TEAC,Letcher2019} with $\calH$ being potential games.
While we also employ game decomposition technique, our target is in the opposite end of those of the cited works above,
which is to seek broad families of games where popular learning algorithms are Lyapunov chaotic.


\parabold{Our Contributions.}
Our starting point is the recent works of Cheung and Piliouras~\cite{CP2019,CP2020}.
They considered a classical technique in the study of dynamical systems, called volume analysis.
Volume analysis considers a set of starting points of positive volume (e.g.~a ball centred at a point).
When this set of starting points evolves according to the rule of dynamical system, it evolves to a new set with a different volume.
Intuitively, volume is a measure of the range of possible outcomes, so the larger it is, the more unpredictable the system is.
Cheung and Piliouras observed that if the set's volume increases exponentially, then its diameter increases exponentially too, which implies Lyapunov chaos.

Cheung and Piliouras employed the volume-expansion argument to show that
Multiplicative Weights Update (MWU) and Follow-the-Regularized-Leader (FTRL) algorithms in two-player zero-sum
and graphical constant-sum games are Lyapunov chaotic everywhere\footnote{By \emph{everywhere}, it means the results hold in any bounded region of the dual space,
for any sufficiently small step-size of the algorithm.
There is also a very mild requirement on \emph{non-triviality} of the game.}
in the cumulative payoff (dual) space; analogous result holds for the optimistic variant of MWU (OMWU) in two-player coordination games.
This indicates that when players repeatedly play the game by using the learning algorithms,
when the initiating condition is slightly perturbed, the cumulative payoffs exhibit a wide range of possibilities in the long run.
This implies instability in the mixed strategy (primal) space.

The volume-expansion argument crucially relies on analyzing the sign of a function $C(\cdot)$ in the dual space.
Volume-expansion with MWU/FTRL (resp.~OMWU) occurs if and only if $C(\cdot)$ is positive (resp.~negative).
Cheung and Piliouras proved that $C(\cdot)$ is always positive in zero-sum games and $C(\cdot)$ is always negative in coordination games.
It is natural to ask: are these chaos results isolated in the sense that they hold only due to the very specific structures of those games,
or do these chaos results generalize broadly?
In the other way around, we ask the following question:
\begin{center}
\emph{How does the Lyapunov chaos phenomena of learning extend beyond}\\
\emph{two-player zero-sum games, two-player coordination games}\\
\emph{and multi-player graphical constant-sum games?}
\end{center}
We answer the above question affirmatively, both for two-player and multi-player settings.

\smallskip

\noindent \emph{Two-player normal form games.}
For the family of two-player normal-form games (bimatrix games) $\calG$, we present two new techniques to go beyond zero-sum game and coordination games.

\textendash The first technique is the well-known direct-sum decomposition $\calG = \calZ \oplus \calC$, in which every bimatrix game $(\bbA,\bbB)$ is decomposed into the sum of a zero-sum game $(\bbZ,-\bbZ)$ and a coordination game $(\bbC,\bbC)$~\cite{BasarHo1974,Kalai2009}. We show that $C_{\bbG}(\cdot)$, the function that determines if the volume is expanding or not in the game $\bbG$, is a simple sum of the zero-sum game part and the coordination part:  $C_{\bbG}(\cdot) = C_{(\bbZ,-\bbZ)}(\cdot) + C_{(\bbC,\bbC)}(\cdot)$ (Theorem~\ref{lem::game::decomposition}).
Recall 
we have discussed that $C_{(\bbZ,-\bbZ)}(\cdot)$ is always positive (as $(\bbZ,-\bbZ)$ is a zero sum game)
and $C_{(\bbC,\bbC)}(\cdot)$ is always negative (as $(\bbC,\bbC)$ is a coordination game).
Thus, if $C_{(\bbZ,-\bbZ)}(\cdot)$ is always relatively larger than $-C_{(\bbC,\bbC)}(\cdot)$, then we always have volume expansion in $\bbG$.

\textendash The second technique is the trivial matrices (see Definition~\ref{def::trivial::matrix}).  Intuitively, trivial matrices are a set of matrices which do not affect the volume changing behavior of the game:
for any bimatrix game $(\bbA, \bbB)$ and any trivial matrices $\bbT^1$ and $\bbT^2$, $C_{(\bbA, \bbB)}(\cdot) \equiv C_{(\bbA + \bbT^1, \bbB + \bbT^2)}(\cdot)$ (Theorem~\ref{lem::trivial::matrix}).
An immediate application of trivial matrices is for bimatrix potential games~\cite{potgames}. For any bimatrix potential game $\bbP$,
$C_{\bbP}(\cdot)$ is identical to $C_{(\bbC,\bbC)}(\cdot)$ for some coordination game $(\bbC,\bbC)$, which implies OMWU is Lyapunov chaotic everywhere in $\bbP$ (Observation~\ref{obs::potential::bimatrix}).

By using these two techniques, we identify two characterizations of bimatrix games where MWU and FTRL are Lyapunov chaotic almost everywhere (Theorem~\ref{thm:matrix-strict-domination} and Theorem~\ref{thm::linear::program}). 
These new characterizations are based on our new notion of \emph{matrix domination} (see Definition~\ref{def::matrix::domination}),
and a linear program (see Eqn.~\eqref{opt::linear::program}) which is designed to prune out the trivial-matrix projection and keep the residue part minimal.
Such family of games has positive Lebesgue measure in the bimatrix game space, so it is not confined to any proper game subspace\footnote{The family of zero-sum games
and the family of coordination games are proper subspaces of the bimatrix game space. Any proper subspace has Lebesgue measure zero.}.
This provides a justification to the claim that the occurrences of chaos are not only circumstantial, but a rather substantial issue of learning in games.
Analogous result holds for OMWU too.

\smallskip

\noindent \emph{Multi-player normal-form games.} For the family of multi-player ($3$ or more players) games,
we first use an observation in~\cite{CP2019}, coupled with our new findings about bimatrix games discussed before,
to present a new family of graphical games in which MWU is Lyapunov chaotic almost everywhere (Theorem~\ref{thm:new-family-chaos-game}); the new family of games strictly includes all graphical constant-sum games.

To facilitate volume analyses in general normal-form games, we establish their \emph{local equivalence of volume change}
with graphical games.
Precisely, we show that $C_{\bbG}(\bbp)$ for a general game $\bbG$ is the same as $C_{\bbH}(\bbp)$ for some graphical game $\bbH$;
$\bbH$ will depend on the point $\bbp$, that's why we say the equivalence is \emph{local} (Theorem~\ref{thm:equivalence}).
This provides an intuitive procedure for understanding volume changes.
Additionally, we show that the volume-changing behaviors of MWU and OMWU are opposite to each other in multi-player game (Proposition~\ref{pr:opposite-multi}).
We use these to analyze MWU and OMWU in multi-player potential games;
in particular, we show that $C_{\bbG}(\bbp)$ of a multi-player potential games is equal to $C_{\bbC}(\bbp)$ of a corresponding multi-palyer coordination game, while $C_{\bbC}(\bbp) \leq 0$ (Lemma~\ref{lem::multi::potential}).



\parabold{Further Related Work.}
Volume analysis has long been a technique of interest in the study of population and game dynamics.
It was discussed in a number of famous texts; see Hofbauer and Sigmund~\cite[Section 11]{HS1998}, Fudenberg and Levine~\cite[Section 3]{Fudenberg98} and Sandholm~\cite[Chapter 9]{Sandholm2010}.
For a modern overview of online learning algorithms from Machine Learning or Economics perspectives,
which includes the discussion about MWU and its variants, no-regret learning and potential games,
we recommend the texts of Cesa-Bianchi and Lugosi~\cite{Cesa06} and Hart and Mas-Collel~\cite{HM2013}.

In the study of no-regret learning (e.g.~\cite{littlestone1994weighted,FS1995}), a vast literature concerns general or even adversarial settings,
in which the online arrivals of payoff values come with no pattern or even from an adversary.
More recently, settings where the online payoffs are more well-behaved, under the term of ``predictable sequence'' coined by Rakhlin and Sridharan~\cite{RakhlinS13}, have been studied.
These settings include game dynamics, as the online payoffs are determined by the mixed strategy choices of the players, while these choices are updated gradually and somewhat predictably.
For these settings, online learning algorithms that perform particularly well, e.g.~achieving regret bound below the canonical $\calO(\sqrt{T})$ limit,
are designed and studied~\cite{HazanKale10,Chiang12,Syrgkanis15}.
For instance, Nesterov's excessive gap technique and optimistic mirror descent are found to achieve near-optimal regret $\calO(\log T)$ in zero-sum games~\cite{DDK2015,RS2013},
and thus the empirical average of the learning sequence converges to Nash equilibrium of the game (see Freund and Schapire~\cite{FS1996} for an explanation).
OMWU (with time-varying step-sizes), and more generally optimistic variant of FTRL~\cite{RakhlinS13}, are some canonical examples of such online learning algorithms.

Recently, there is a stream of work that examines how learning algorithms behave in games or min-max optimization from a dynamical-systemic perspective.
Replicator dynamics (RD; the continuous-time analogue of MWU) and continuous-time FTRL are found to achieve optimal regret in general settings~\cite{MPP2018}.
Furthermore, RD in zero-sum games or graphical constant-sum games admits a constant of motion and preserves volume;
these two properties are used to show that such dynamical systems are near-periodic~\cite{HS1998,PS2014,MPP2018,boone2019darwin},
captured rigorously under the notion of \emph{Poincar\'{e} recurrence}~\cite{Poincare1890,barreira}.
However, when MWU, the forward Euler discretization of RD, is used in discrete-time setting in zero-sum games, the near-periodicity is destroyed totally;
indeed, the system will never visit the same point (or its tiny neighbourhood) twice, converge to the boundary of the strategy simplex, and fluctuate there irregularly~\cite{BP2018,Cheung2018}.

In contrast, (discrete-time) OMWU in zero-sum game is shown to converge to Nash equilibrium~\cite{daskalakis2018last}; yet, in the more general setting of min-max optimization,
it was found that Optimistic Gradient Descent Ascent (OGDA) can have limit points other than (local) min-max solutions~\cite{daskalakis2018limit}.

Another notion of chaos called \emph{Li-Yorke chaos} was shown to exist when a variant of MWU is used in congestion games~\cite{PPP2017}.



\section{Preliminary}\label{sect:prelim}

In this paper, every bold lower-case alphabet denotes a vector, every bold upper-case alphabet denotes a matrix or a game.
When we say a ``game'', we always mean a normal-form game.
Given $n$, let $\Delta^n$ denote the mixed strategy space of dimension $n$, i.e.~$\{(z_1,z_2,\cdots,z_n) ~|~ \sum_{j=1}^n z_j = 1\}$.

%
%

\parabold{Normal-Form Games.}
We use $N$ to denote the number of players of a game. Let $S_i$ denote the strategy set of Player $i$, and $S := S_1 \times S_2\times \cdots \times S_N$.
Let $n_i = |S_i|$. $\bbs = (s_1,s_2,\cdots,s_N) \in S$ denotes a strategy profile of all players, and $u_i(\bbs)$ denotes the payoff to Player $i$ when each player picks $s_i$.
A mixed strategy profile is denoted by $\bbx = (\bbx_1,\bbx_2,\cdots,\bbx_N)$,
and $u_i$ is extended to take mixed strategies as inputs via $u_i(\bbx) = \expectsub{\bbs\sim \bbx}{u_i(\bbs)}$.
Also, we let
\begin{align}
U^{i_1 i_2 \cdots i_g}_{j_1 j_2 \cdots j_g}(\bbx) =~ & \text{the expected payoff to Player $i_1$ when: for $1\le f \le g$, Player $i_f$ picks strategy $j_f$,}\nonumber\\
& \text{while for each player $i \notin \{i_1, \cdots i_g\}$, she picks a strategy randomly following $\bbx_i$}\nonumber\\
=~& \expectsub{\bbs_{-(i_1,\cdots,i_g)}\sim \bbx_{-(i_1,\cdots,i_g)}}{u_{i_1}(s_{i_1}=j_1,\cdots,s_{i_g}=j_g,\bbs_{-(i_1,\cdots,i_g)})}.\label{eq:U}
\end{align}
Note that $-(i_1, \cdots, i_g)$ denotes the player set other than $i_1, \cdots, i_g$. Also, we use $U^{i_1 i_2 \cdots i_g}_{j_1 j_2 \cdots j_g}$ if $\bbx$ is clear from the context.
We say a game is a zero-sum game if $\sum_i u_i(\bbs) = 0$ for all $\bbs\in S$, and we say a game is a coordination game if $u_i(\bbs) = u_k(\bbs)$ for all Players $i$ and $k$ and for all $\bbs \in S$.


When $N=2$, such games are called \emph{bimatrix games}, for which we adopt the notations below.
Let $(\bbA,\bbB)$ denote a bimatrix game, where for any $j\in S_1$, $k\in S_2$, $A_{jk} := u_1(j,k)$, $B_{jk} := u_2(j,k)$.
$\bbx$ and $\bby$ denote mixed strategies of Players 1 and 2 respectively.
A bimatrix game is a zero-sum game if $\bbA = -\bbB$; it is a coordination game if $\bbA = \bbB$.
Note that $U^1_j = [\bbA \bby]_j$, $U^2_k = [\bbB\trans\bbx]_k$, which we denote by $A_j,B_k$ respectively when $\bbx,\bby$ are clear from context;
$B_j,A_k$ are defined analogously.

\parabold{MWU, FTRL and OMWU in Games.}
All three algorithms have a step-size $\ep$, and can be implemented as updating in the cumulative payoff (dual) space.
In each round, the players' actions (mixed strategies) in the primal space are functions of the cumulative payoff vectors to be defined below,
and these actions are then used to determine the payoffs in the next round.
For a player with $d$ strategies, let $\bbp^t\in \rr^d$ denote her cumulative payoff vector at time $t$,
and let $\bbp^0\in \rr^d$ denote the starting point chosen by the player.
For MWU in a game, the update rule for Player $i$ is
\begin{equation}\label{eq:MWU}
p_j^{t+1} ~=~ p_j^t ~+~ \ep \cdot U^i_j(\bbx^t),
\end{equation}
where $U^i_j$ is the function defined in~\eqref{eq:U}, and $\bbx^t$ is the mixed strategy determined by the formula below:
\begin{equation}\label{eq:dual-to-primal}
x_j^t ~=~ x_j(\bbp^t) ~=~ \exp(p_j^t)/({\textstyle \sum_{\ell\in S_i}} \exp(p_\ell^t))
\end{equation}

For OMWU in a game, the update rule for Player $i$ starts with $\bbp^1 = \bbp^0$, and for $t\ge 2$,
\[
p_j^{t+1} = p_j^t + \ep \cdot \left[2 U^i_j(\bbx^t) - U^i_j(\bbx^{t-1})\right],
\]
where $\bbx^t$ is determined by~\eqref{eq:dual-to-primal}.

For FTRL in a game, the update rule for Player $i$ is same as~\eqref{eq:MWU}, but $\bbx^t$ is determined as below using
a convex \emph{regularizer function} $h_i: \Delta^d \ra \rr$: $\bbx^t = \argmax_{\bbx\in \Delta^d} \left\{ \inner{\bbp^t}{\bbx} - h_i(\bbx) \right\}$.
As all the results for MWU can be directly generalized to FTRL as discussed in~\cite[Appendix D]{CP2019}, to keep our exposition simple, in the rest of this paper,
we focus on MWU and OMWU and their comparisons.
For bimatrix game, 
we use $\bbp,\bbq$ to denote the cumulative payoff vectors of Players 1 and 2 respectively.

\parabold{Dynamical Systems, Lyapunov Chaos and Volume Analysis.}
A learning-in-game system can be viewed as a discrete-time dynamical system.
We present a simplified definition of dynamical systems that fits our need.
A discrete-time dynamical system in $\rr^d$ is determined by a starting point $\bbs(0)\in \rr^d$ and an update rule $\bbs(t+1) = f(\bbs(t))$, where $f:\rr^d \ra \rr^d$ is
a function.\footnote{OMWU in game is not a dynamical system, as the update to $\bbs(t+1)$ depends on both $\bbs(t),\bbs(t-1)$.
But there is a function $f$ such that $\bbs(t+1)\approx f(\bbs(t))$, while the volume-changing behavior is not really affected~\cite{CP2020}.}
The sequence $\bbs(0),\bbs(1),\bbs(2),\cdots$ is called a \emph{trajectory} of the dynamical system.
When $f$ is clear from the context, we let $\Phi:(\nn \cup \{0\}) \times \rr^d \ra \rr^d$ denote the function
such that $\Phi(t,\bbs)$ is the value of $\bbs(t)$ generated by the dynamical system with starting point being $\bbs$.
Given a set $U\subset \rr^d$, we let $\Phi(t,U) = \{ \Phi(t,\bbs) | \bbs \in U \}$.
Let $B(\bbs,r)$ denote the open ball with center $\bbs$ and radius $r$.

There are a number of similar but not identical definitions of Lyapunov chaos, all capturing the \emph{butterfly effect}:
when the starting point is slightly perturbed, the resulting trajectories diverge quickly.
We use the following definition, which was also used in~\cite{CP2019,CP2020} implicitly.
Intuitively, a system is Lyapunov chaotic in an open set $S$ if for any $\bbs\in S$ and any open ball $B$ around $\bbs$,
as long as $\Phi(t,B)$ remains inside $S$, there exists $\bbs' \in B$ such that $\|\Phi(t,\bbs') - \Phi(t,\bbs)\|$ grows exponentially with $t$.
Lyapunov exponent in the definition is a measure of how fast the exponential growth is; the larger it is, the more unpredictable the dynamical system is.

\begin{definition}\label{def:lyapunov-chaos}
We say a dynamical system is \emph{Lyapunov chaotic} in an open set $S \subset \rr^d$ if there exists a constant $\lambda > 0$ and a \emph{Lyapunov exponent} $\gamma = \gamma(S) > 0$,
such that for any $\bbs\in S$, for any sufficiently small $\delta > 0$ and for all $t$ satisfying $0\le t < \min \{\tau| \Phi(\tau,B(\bbs,\delta))\subsetneq S \}$,
\[
\sup_{\bbs'\in B(\bbs,\delta)} ~~\|\Phi(t,\bbs') - \Phi(t,\bbs)\| ~~\ge~~ \lambda \cdot \delta \cdot \exp(\gamma t).
\]
We say a dynamical system is \emph{Lyapunov chaotic everywhere} if it is Lyapunov chaotic in any bounded open set $S\subset \rr^d$.
\end{definition}

In the above definition, all norms and radii are Euclidean norms. For capturing round-off errors in computer simulations and ML systems,
it is more natural to use $\ell_1$-norm for which $\delta$ is the round-off maximum error, say $\sim 10^{-16}$ when IEEE 754 binary64 (standard double) is used.

When $S$ is a small set, it is usually easy to determine whether a dynamical system is Lyapunov chaotic in $S$, since the dynamic can be locally approximated by a linear dynamical system,
where the eigenvalues of the local Jacobian characterizes chaotic behaviors (when $f$ is smooth).
But when $S$ is a large, determining whether Lyapunov chaos occurs is difficult in general.
Cheung and Piliouras~\cite{CP2019} found that 
volume analysis can be useful in this regard, based on the following simple observation.

\begin{proposition}\label{pr:vol-to-diameter}
In $\rr^d$, if a set $S$ has volume at least $v$, then the radius w.r.t.~any point $\bbs \in S$ is at least $v^{1/d}/2$.
Thus, if the volume of $\Phi(t,S)$ of some dynamical system is $\Omega(\exp(\gamma t))$ for some $\lambda,\gamma > 0$,
then the radius of $S(t)$ w.r.t.~any point $\bbs\in S(t)$ is $\Omega(\exp(\frac \gamma d\cdot  t))$.
\end{proposition}

Cheung and Piliouras showed Lemma~\ref{lem:chaos} below, which, for bimatrix games, reduces volume analysis to analyzing the sign the function
$C_{(\bbA,\bbB)}(\bbp,\bbq)$ defined in Eqn.~\eqref{eq:Cxy} below; the sign also 
determines the local volume-changing behavior around the point $(\bbp, \bbq)$ when MWU is used.
Based on Proposition~\ref{pr:vol-to-diameter} that converts volume expansion to radius expansion, the sign 
can be used to determine if the dynamical system is Lyapunov chaotic.
In Eqn.~\eqref{eq:Cxy}, $\bbx(\bbp),\bby(\bbq)$ are mixed strategies of Players 1 and 2 respectively, computed using~\eqref{eq:dual-to-primal}.
Equality~\eqref{eq:Cxy-expect} can be derived easily, in which the expectation $\expect{\cdot}$ is indeed $\expectsub{(j,k)\sim (\bbx(\bbp),\bby(\bbq))}{\cdot}$,
i.e.~the underlying distribution is where $j$ is drawn following the distribution $\bbx(\bbp)$, while $k$ is drawn following the distribution $\bby(\bbq)$.
\begin{align}
C_{(\bbA,\bbB)}(\bbp,\bbq) 
&=~ -\sum_{j\in S_1}~\sum_{k\in S_2}~x_j(\bbp) \cdot y_k(\bbq) \cdot
(A_{jk} - [\bbA\cdot \bby(\bbq)]_j) \cdot (B_{jk} - [\bbB\trans \cdot \bbx(\bbp)]_k) \label{eq:Cxy}\\
&=~ - \expect{(A_{jk}-A_j - A_k)(B_{jk} - B_j - B_k)} + \expect{A_{jk}} \cdot \expect{B_{jk}}.\label{eq:Cxy-expect}
\end{align}

For multi-player game $\bbG$, the analogous function $C_\bbG(\cdot)$ is given below; the $U$ quantities were defined in~\eqref{eq:U}.
Lemma~\ref{lem:chaos} is adapted from \cite{CP2019} for games with any number of players.
Derivation of~\eqref{eq:Cxy-multi} uses the Jacobian of the corresponding dynamical system and integration by substitution; see Appendix~\ref{app:equivalence}.
\begin{equation}\label{eq:Cxy-multi}
C_\bbG(\bbp_1,\cdots,\bbp_N) ~=~ - \sum_{i\in [N],~j\in S_i} ~\sum_{k > i,~\ell\in S_k} ~x_{ij} x_{k\ell}
\left( U^{ki}_{\ell j} - U^k_\ell \right) \left( U^{ik}_{j\ell} - U^i_j \right).
\end{equation}

\begin{lemma}\label{lem:chaos}
Let $\bbG$ be a game.
Suppose that $S$ is a 
set in the dual space $\rr^d$, and
\begin{equation}\label{eq:barc}
\bar{c}(S) := \inf_{(\bbp_1,\cdots,\bbp_N)\in S} C_{\bbG}(\bbp_1,\cdots,\bbp_N) > 0.
\end{equation}
Then for MWU in the bimatrix game with any sufficiently small step-size $\ep$,
as long as $\Phi(t,B)\subset S$ for all $0\le t \le \overline{T}$, then for all $t$ in this range,
the volume of $\Phi(t,B)$ is at least $\Phi(0,B)\cdot \exp(\frac{\bar{c}(S)}{2}\ep^2\cdot t)$,
and hence the radius of $\Phi(t,B)$ is at least $\Omega(\exp(\frac{\bar{c}(S)}{2d}\ep^2\cdot t))$.
Subsequently, the dynamical system is Lyapunov chaotic in $S$ with Lyapunov exponent $\frac{\bar{c}(S)}{2d}\ep^2$.

If MWU is replaced by OMWU, then the same result holds by replacing the condition~\eqref{eq:barc} with $\bar{c}(S) := \inf_{(\bbp_1,\cdots,\bbp_N)\in S} [-C_{\bbG}(\bbp_1,\cdots,\bbp_N)] > 0$.
\end{lemma}

Note that if we start from a Nash equilibrium in the primal space, MWU and OMWU 
will stay at the equilibrium. However, if this equilibrium $(\bbx^*,\bby^*)$ satisfies the conditions in Corollary~\eqref{coro:escape} below,
there are points arbitrarily close to the equilibrium that keep moving away from the equilibrium (if the region $\{(\bbx',\bby') = (\bbx(\bbp),\bby(\bbq)) | (\bbp,\bbq)\in S\}$ is large).

\begin{corollary}[Adapted from{~\cite[Theorem 5]{CP2020}}]\label{coro:escape}
Let $(\bbx^*,\bby^*)$ be a point in the interior of the primal space. Suppose that there exists $(\bbp,\bbq)$ in the dual space, such that
$\bbx^* = \bbx(\bbp)$ and $\bby^* = \bby(\bbq)$. Furthermore, suppose $C_{(\bbA,\bbB)}(\bbp,\bbq) > 0$ and $(\bbp,\bbq)\in S$ where $S$ is the set described in Lemma~\ref{lem:chaos}.
Then there are primal points arbitrarily close to $(\bbx^*,\bby^*)$ such that MWU in the game $(\bbA,\bbB)$ eventually leaves the corresponding primal set of $S$,
i.e.~$\{(\bbx',\bby') = (\bbx(\bbp),\bby(\bbq)) | (\bbp,\bbq)\in S\}$.
\end{corollary}

When the game is zero-sum, i.e., $\bbB = - \bbA$, hence $C_{(\bbA,\bbB)}(\bbp,\bbq) = \expect{(A_{jk}-A_j-A_k)^2} - \expect{A_{jk}}^2$.
Since $\expect{A_{jk}} = \expect{A_j} = \expect{A_k}$ and hence $\expect{A_{jk}-A_j-A_k} = - \expect{A_{jk}}$,
$C_{(\bbA,\bbB)}(\bbp,\bbq)$ is indeed the variance of the random variable $A_{jk}-A_j-A_k$,
and thus is non-negative.

By~\eqref{eq:Cxy}, we have $C_{(\bbA,\bbB)}(\bbp,\bbq) = -C_{(\bbA,-\bbB)}(\bbp,\bbq)$. Thus, for any coordination game $(\bbA,\bbA)$,
$C_{(\bbA,\bbA)}(\bbp,\bbq) = -C_{(\bbA,-\bbA)}(\bbp,\bbq)\le 0$, due to the observation about zero-sum games above.


\section{Bimatrix Games}

\label{sec::bimatrix}
In this section, we focus on general bimatrix games $(\bbA, \bbB)$. First, in Section~\ref{sec::tools::bimatrix}, we present two tools for analyzing $C_{(\bbA, \bbB)}(\cdot)$,
and then we provide an example to show how to use these two tools. Finally, in Section~\ref{sec::results::bimatrix}, we present two characterizations such that the dynamics are Lyapunov chaotic almost everywhere.

\subsection{Tools for Analyzing Bimatrix Game} \label{sec::tools::bimatrix}

\parabold{First Tool: Canonical Decomposition for Bimatrix Games.}
For every bimatrix game $(\bbA, \bbB)$, it admits a canonical decomposition~\cite{BasarHo1974,Kalai2009} into the sum of a zero-sum game $(\bbZ, -\bbZ)$ and a coordination game $(\bbC, \bbC)$,
where $\bbZ = \frac{1}{2} (\bbA - \bbB)$ and $\bbC = \frac{1}{2} (\bbA + \bbB)$, i.e.
\[
(\bbA, \bbB) = (\bbZ, -\bbZ) + (\bbC, \bbC).
\]
We call $(\bbZ, -\bbZ)$ the zero-sum part of the game $(\bbA, \bbB)$, and $(\bbC, \bbC)$ the coordination part of the game.
Our first result shows that the function $C(\cdot)$ can be decomposed neatly into the two parts too.

\begin{lemma}\label{lem::game::decomposition}
  For any bimatrix game $(\bbA, \bbB)$,
\begin{align*}
C_{(\bbA,\bbB)}(\bbp,\bbq) ~\equiv~ C_{(\bbZ,-\bbZ)}(\bbp,\bbq) + C_{(\bbC,\bbC)}(\bbp,\bbq),
\end{align*}
where $\bbZ = \frac{1}{2} (\bbA - \bbB)$ and $\bbC = \frac{1}{2} (\bbA + \bbB)$.
\end{lemma}

\begin{proof}
We use~\eqref{eq:Cxy-expect} to expand the following:
\begin{align*}
&4 \cdot C_{(\bbZ,-\bbZ)}(\bbp,\bbq) + 4 \cdot C_{(\bbC,\bbC)}(\bbp,\bbq) \\
=& \expect{(A_{jk}-B_{jk}-A_j+B_j-A_k+B_k)^2} - \expect{A_{jk}-B_{jk}}^2\\
&\qquad\qquad ~-~ \expect{(A_{jk}+B_{jk}-A_j-B_j-A_k-B_k)^2} + \expect{A_{jk}+B_{jk}}^2\\
=& \expect{(A_{jk}-B_{jk}-A_j+B_j-A_k+B_k)^2 - (A_{jk}+B_{jk}-A_j-B_j-A_k-B_k)^2}\\
&\qquad\qquad - \left( \expect{A_{jk}} - \expect{B_{jk}} \right)^2 + \left( \expect{A_{jk}} + \expect{B_{jk}} \right)^2\\
=& \expect{4(-B_{jk}+B_j+B_k)(A_{jk}-A_j-A_k)} + 4 \cdot \expect{A_{jk}}\cdot \expect{B_{jk}}\\
=& 4\cdot C_{(\bbA,\bbB)}(\bbp,\bbq).\qedhere
\end{align*}
\end{proof}

By the end of Section~\ref{sect:prelim}, we discussed that for any zero-sum game, $C_{(\bbZ,-\bbZ)}(\bbp,\bbq)$ is always non-negative,
and for any coordinate game, $C_{(\bbC,\bbC)}(\bbp,\bbq)$ is always non-positive.
By using the above lemma, we can analyze the volume-changing behavior of a bimatrix game $(\bbA, \bbB)$ by looking at its zero-sum and coordination parts independently.
One simple intuition is that 
if the coordination (resp.~zero-sum) part is small, then the volume-changing behavior of $(\bbA, \bbB)$ is closer to the behavior of the zero-sum (resp.~coordination) part.\vspace*{-0.02in}
We realize this intuition quantitatively in the next subsection.


\parabold{Second Tool: Trivial matrix.} Trivial matrices are matrices which do not affect the volume-changing behavior, as depicted in Lemma~\ref{lem::trivial::matrix} below.
\begin{definition}[Trivial Matrix]\label{def::trivial::matrix}
$\bbT \in \rr^{n\times m}$ is a \emph{trivial matrix} if there exists real numbers $u_1,u_2,\cdots,u_n$ and $v_1,v_2,\cdots,v_m$
such that $T_{jk} = u_j + v_k$ for all $j\in [n],k\in [m]$.
\end{definition}

\begin{lemma} \label{lem::trivial::matrix}
For any two trivial matrices $\bbT^1,\bbT^2$, for any two matrices $\bbA,\bbB\in \rr^{n\times m}$,
\[
C_{(\bbA,\bbB)}(\bbp,\bbq) ~\equiv~ C_{(\bbA+\bbT^1,\bbB+\bbT^2)}(\bbp,\bbq).
\]
\end{lemma}

One immediate application of this lemma is for two player potential games.
\begin{definition}\label{def::potential::multi}
A game $\bbG$ is a potential game if there exists a potential function $\mathcal{P}:S\ra \rr$ such that for any Player $i$ and any strategy profile $\bbs\in S$,
\begin{align*}
  \mathcal{P}(s_i, \bbs_{-i}) - \mathcal{P}(s'_i, \bbs_{-i}) = u_i(s_i, \bbs_{-i}) - u_i(s_{i'}, \bbs_{-i}).
\end{align*}
\end{definition}

For the potential game, we have the following observation:
\begin{obs}\label{obs::potential::bimatrix}
For any bimatrix potential game $(\bbA,\bbB)$, there is a coordination game $(\bbP, \bbP)$ such that $\bbA - \bbP,~\bbB - \bbP$ are trivial matrices.
$\bbP$ is the matrix representation of the potential function $\mathcal{P}$.
\end{obs}
This observation immediately implies that the volume-changing behavior of potential game is equivalent to that of a corresponding coordination game.

We give a concrete example to show how these tools help us to analyze the $C_{(\bbA, \bbB)}(\cdot)$.\vspace*{-0.03in}

\parabold{A Simple Example.}
We will show how to use our tools to demonstrate $C(\cdot) \geq 0$ everywhere for the following game.
In the example, each player has three strategies. The payoff bimatrix $(\bbA,\bbB)$ is given below.
The first number gives the payoff of the row player, who chooses strategy from $\{a, b, c\}$; the second number gives the payoff of the column player, who chooses strategy from $\{1, 2, 3\}$.
{\small
\begin{table}[htp]
  \centering
\begin{tabular}{|c|c|c|c|}
\hline
& Strategy $1$  & Strategy $2$  & Strategy $3$ \\
\hline
Strategy $a$ & $(4,4)$ & $(12, -4)$ & $(-6, 10)$  \\
\hline
Strategy $b$ & $(-8, 8)$ & $(0, 0)$ & $(12, -4)$  \\
\hline
Strategy $c$ & $(14, -2)$ & $(-8, 8)$ & $(4, 4)$ \\
\hline
\end{tabular}
\end{table}
}
We first use our \emph{first tool} to decompose this game into zero-sum part $(\bbZ,-\bbZ)$ and coordination part $(\bbC,\bbC)$, where
$\bbZ = \left[\begin{smallmatrix}
0 & 8 & -8\\
-8 & 0 & 8 \\
8 & 0 & 0
\end{smallmatrix}\right]$ and $\bbC = \left[\begin{smallmatrix}
4 & 4 & 2\\
0 & 0 & 4 \\
6 & 0 & 4
\end{smallmatrix}\right]$.
At this point, we still cannot easily figure out which one is larger between $C_{(\bbZ, -\bbZ)}(\cdot)$ and $C_{(\bbC, -\bbC)})(\cdot)$.
However, we can further decompose the coordination part by the \emph{second tool}: $\bbC = \left[\begin{smallmatrix}
  4 & 2 & 4\\
  2 & 0 & 2 \\
  4 & 2 & 4
  \end{smallmatrix}\right] +
  \left[\begin{smallmatrix}
  0 & 2 & -2\\
  -2 & 0 & 2 \\
  2 & -2 & 0
  \end{smallmatrix}\right]$, where the first matrix on the RHS is a trivial matrix.
It's easy to see 
the second matrix on the RHS is $\frac{1}{4} \bbZ$.
Then by Lemmas~\ref{lem::game::decomposition} and~\ref{lem::trivial::matrix}, and the definition of the function $C$, for any point $(\bbp,\bbq)$ in the dual space,
\begin{align*}
  C_{(\bbA, \bbB)}(\bbp,\bbq) = C_{(\bbZ, -\bbZ)}(\bbp,\bbq) + C_{(\frac{1}{4}\bbZ, \frac{1}{4}\bbZ)}(\bbp,\bbq) = \left(1 - 
  (1/4)^2\right)\cdot C_{(\bbZ, -\bbZ)}(\bbp,\bbq) \geq 0.
\end{align*}

\subsection{Results for Bimatrix Games}\label{sec::results::bimatrix}
In this subsection, we identify several characterizations for general bimatrix games in which we have chaotic behavior with MWU dynamic in a following set $S$ in the cumulative payoff (dual) space $\mathbb{R}^{n_1+n_2}$:
\begin{align*}
  S^{\delta} = \left\{ (\bbp, \bbq)  \left| \forall j\in S_1, k\in S_2, ~~ x_j(\bbp) \geq \delta ~~\texttt{and}~~ y_k(\bbq) \geq \delta\right.\right\}.
\end{align*}

In order to show chaotic behavior of MWU in a specific bimatrix game $(\bbA, \bbB)$,
it is sufficient to show $C_{(\bbA, \bbB)}(\bbp, \bbq)$ is strictly positive in the region $S^{\delta}$, followed by applying Lemma~\ref{lem:chaos}.
In 
the previous subsection, we show that for each game $(\bbA, \bbB)$, it can be decomposed into a zero-sum part $(\bbZ, -\bbZ)$ and a coordination part $(\bbC, \bbC)$.
Furthermore, $C_{(\bbA, \bbB)}(\bbp, \bbq) = C_{(\bbZ, -\bbZ)}(\bbp, \bbq) + C_{(\bbC, \bbC)}(\bbp, \bbq)$.
We also raise an intuition that if the zero-sum part is small, then the volume behavior in the game $(\bbA, \bbB)$ will be similar that in the coordination part;
conversely, if the coordination part is small, then the volume behavior will be similar to the zero-sum part.
However, we have not yet presented a way to compare the largeness of the two parts. This is what we do here.

\subsubsection{First Characterization: Matrix Domination}

The first characterization we identify is matrix domination. In this part, we show that under certain conditions,
the zero-sum part is always no less than the coordination part,
i.e.~$C_{(\bbZ, -\bbZ)}(\bbp, \bbq) \geq -C_{(\bbC, \bbC)}(\bbp, \bbq)$ for all $(\bbp,\bbq)$.
This directly implies $C_{(\bbA, \bbB)}(\bbp, \bbq)$ will be non-negative in the whole dual space.
Interestingly, the condition we identify is both necessary and sufficient.
Similar result can also be achieved in the case that coordination part is always no less than the zero-sum part. 
We first introduce the definition of the matrix domination.

\begin{definition}\label{def::matrix::domination}
We say matrix $\bbK$ dominates matrix $\bbL$ if they are of the same dimension, and for any row indices $j,j'$ and column indices $k,k'$,
\[
\left|\bbK_{jk} + \bbK_{j'k'} - \bbK_{jk'} - \bbK_{j'k}\right| \geq \left|\bbL_{jk} + \bbL_{j'k'} - \bbL_{jk'} - \bbL_{j'k}\right|.
\]
\end{definition}
Note that the domination induces a partial order on all matrices: if $\bbK$ dominates $\bbL$ and $\bbL$ dominates $\bbM$, then $\bbK$ dominates $\bbM$.
The theorem below gives the necessary and sufficient condition.
\begin{theorem}\label{thm:matrix-domination}
  $C_{(\bbA, \bbB)}(\bbp, \bbq)$ is non-negative for all $\bbp$ and $\bbq$ if and only if matrix of the zero-sum part $\bbZ$ dominates the coordination part $\bbC$.
\end{theorem}

The above theorem is based on the following crucial observation.
\begin{obs}\label{obs::C::dec}
For any matrix $\bbZ$,
\begin{align*}
C_{(\bbZ, -\bbZ)}(\bbp, \bbq) = \frac{1}{4} \sum_{\substack{j, j'\in S_1\\k, k'\in S_2}} x_j(\bbp) \cdot y_k(\bbq) \cdot x_{j'}(\bbp) \cdot y_{k'}(\bbq) \cdot \left( Z_{jk} + Z_{j'k'} - Z_{jk'} - Z_{j'k}\right)^2.
\end{align*}
\end{obs}

Matrix domination only implies $C_{(\bbA, \bbB)}(\bbp, \bbq)$ is non-negative. In order to have $C_{(\bbA, \bbB)}(\bbp, \bbq)$ to be strictly positive in the set $S$, we need $\theta$-domination.
\begin{definition}\label{def::theta::domination}
We say matrix $\bbK$ $\theta$-dominates ($\theta > 0$) matrix $\bbL$ if
$\bbK$ dominates $\bbL$, and there exist $j$, $j'$, $k$, $k'$ such that
\[
\left|\bbK_{jk} + \bbK_{j'k'} - \bbK_{jk'} - \bbK_{j'k}\right| \geq \left|\bbL_{jk} + \bbL_{j'k'} - \bbL_{jk'} - \bbL_{j'k}\right| + \theta.
\]
\end{definition}

The following theorem holds due to Lemma~\ref{lem:chaos}.

\begin{theorem} \label{thm:matrix-strict-domination}
  For any general bimatrix game $(\bbA, \bbB)$ which is decomposed into zero-sum part $(\bbZ, -\bbZ)$ and coordination part $(\bbC, \bbC)$,
  if $\bbZ$ $\theta$-dominates $\bbC$, then MWU with any sufficiently small step-size $\ep$ in the game $(\bbA,\bbB)$
is Lyapunov chaotic in $S^\delta$ with Lyapunov exponent $\frac{\theta^2 \delta^2}{2(n_1+n_2) }\ep^2$.
\end{theorem} 

Note that in Definition~\ref{def::theta::domination}, $\bbK$ $\theta$-dominates $\bbL$ if a finite number of inequalities are satisfied.
In the context of Theorem~\ref{thm:matrix-strict-domination}, it is easy to see that there is quite many games $(\bbA,\bbB)$,
such that $\bbZ$ $\theta$-dominates $\bbC$ with all those inequalities \emph{strictly} satisfied.
Thus, there exists an open neighbourhood around these games
such that every game in the neighbourhood has its zero-sum part $\theta$-dominates its coordination part.
This shows that such family of games has positive Lebesgue measure.

\subsubsection{Second Characterization: Linear Program}

Note that matrix domination is not always true. In some scenarios, the zero-sum matrix might not dominate the coordination matrix.
Yet, it is still possible that $C_{(\bbA, \bbB)}(\bbp, \bbq)$ is strictly positive in the region
$S^{\delta} = \left\{ (\bbp, \bbq)  \left| \forall j, k ~ x_j(\bbp) \geq \delta ~\texttt{and}~ y_k(\bbq) \geq \delta\right.\right\}$,
when every 
entry in the coordination matrix is \emph{small}.

Precisely, for a general bimatrix game $(\bbA, \bbB)$, if its coordination part $(\bbC, \bbC)$ is small in the sense that the absolute values of all entries in $\bbC$ are smaller than some constant $r$, then we can bound $C_{(\bbC, -\bbC)}(\cdot)$ by $\calO(r^2)$. This is not the only case we can bound $C_{(\bbC, -\bbC)}(\cdot)$ by a small term. Even the entries in matrix $\bbC$ are large, we can use trivial matrices to reduce them 
without affecting $C_{(\bbC, \bbC)}(\cdot)$. This is done via a linear programming approach described below.

Given a matrix $\bbK$, let $r(\bbK)$ be the optimal value of following linear program:
\begin{equation}\label{opt::linear::program}
\min_{r\ge 0}~~r~~\text{such that}~~\forall j,k,~~-r \le K_{jk} - g_j - h_k \le r.
\end{equation}
Note that $\{g_j + h_k\}_{j,k}$ constructs a trivial matrix. Let $\bbK' = \bbK - \{g_j + h_k\}_{j,k}$.
By Lemma~\ref{lem::trivial::matrix}, $C_{(\bbK, -\bbK)}(\cdot) = C_{(\bbK', -\bbK')}(\cdot)$.
The following lemma shows that the value of $C_{(\bbC, -\bbC)}(\cdot)$ is closely related to $r(\bbC)$.

\begin{lemma} \label{lem::linear::program}
For any $(\bbp, \bbq)$ in $S^\delta = \left\{ (\bbp, \bbq)  \left| \forall j, k ~ x_j(\bbp) \geq \delta ~~\texttt{\emph{and}}~~ y_k(\bbq) \geq \delta\right.\right\}$,
\[
(r(\bbC)) \cdot \delta)^2 \leq C_{(\bbC, -\bbC)}(\bbp, \bbq) \leq r(\bbC)^2
\]
\end{lemma}

Then, by Lemma~\ref{lem::linear::program}, the theorem below follows by applying Lemma~\ref{lem:chaos}.
\begin{theorem}\label{thm::linear::program}
For any general bimatrix game $(\bbA, \bbB)$ which is decomposed into zero-sum part $(\bbZ, -\bbZ)$ and coordination part $(\bbC, \bbC)$,
if $\left(r(\bbZ) \delta\right)^2 \ge \left(r(\bbC)\right)^2 + (\theta\delta)^2$, 
then MWU with any sufficiently small step-size $\epsilon$ in the game $(\bbA,\bbB)$
is Lyapunov chaotic in $S^\delta$ with Lyapunov exponent $\frac{\theta^2}{2(n_1+n_2)}\ep^2$.
\end{theorem}



\section{Multi-Player Games}

Computing volume change of learning algorithm in multi-player game is slightly more involved than the two-player case.
We present a local equivalence formula of volume change between normal-form and graphical games.
This provides an intuitive procedure for understanding volume changes.
Proposition~\ref{pr:opposite-multi} 
shows that in multi-player game, the volume-changing behaviors of MWU and OMWU are again \emph{opposite} to each other (which was shown for bimatrix game in~\cite{CP2020}).

\parabold{Graphical Games.} A graphical game~\cite{KearnsLS2001} is a special type of $N$-player game where the payoffs can be compactly represented.
In a graphical game $\bbH$, for each pair of players $i,k$, there is an \emph{edge-game} which is a bimatrix game between the two players, denoted by
$(\bbH^{i,k},(\bbH^{k,i})\trans)$, where $\bbH^{i,k} \in \rr^{n_i \times n_k}$ is the payoff matrix that denotes the payoffs to Player $i$.
Then the payoff to Player $i$ at strategy profile $\bbs = (s_1,s_2,\cdots,s_N)$ is the sum of payoffs to Player $i$ in all her edge-games, i.e.~$u_i(\bbs) ~=~ \sum_{k\neq i} H^{i,k}_{s_i,s_k}$.
As is standard, this payoff function is extended via expectation when the inputs are mixed strategies.

Here, we first use an observation from~\cite{CP2019} 
to construct a family of multi-player graphical games
where MWU is Lyapunov chaotic in $S^{N,\delta} := \{ (\bbp_1,\cdots,\bbp_N) | \forall i\in [N],j\in S_i,~x_{ij}(\bbp_i) \ge \delta \}$. It was observed that the function $C_\bbH(\bbp)$ defined in~\eqref{eq:Cxy-multi} is the sum of $C_{(\bbH^{i,k},(\bbH^{k,i})\trans)}(\bbp_i,\bbp_k)$ of all pairs of Players $i<k$~\cite{CP2019}.
This observation 
yields Theorem~\ref{thm:new-family-chaos-game}.

\begin{theorem}\label{thm:new-family-chaos-game}
Let $\calG^{\uparrow}$ denote the family of bimatrix games which satisfy the condition either in Theorem~\ref{thm:matrix-strict-domination} or in Theorem~\ref{thm::linear::program}.
In an $N$-player graphical game where each edge-game is drawn from $\calG^{\uparrow}$, if all players are employing MWU with a sufficiently small step-size $\ep$,
then the dynamical system is Lyapunov chaotic in $S^{N,\delta}$ with Lyapunov exponent $N(N-1)\theta^2 \delta^2 \ep^2 / 4\sum_{i=1}^N n_i$.
\end{theorem}

\parabold{Local Equivalence of General Games and Graphical Games.}
Next, we present a theorem which connects the value of  $C_\bbG(\bbp)$ of a general game to $C_\bbH(\bbp)$, where $\bbH$ is a graphical game.
\begin{theorem} \label{thm:equivalence}
Given an $N$-player normal-form game $\bbG$ and any point $\bbp$ in the dual space, the value of $C_\bbG(\bbp)$ is the same as $C_\bbH(\bbp)$,
where $\bbH$ is a graphical game specified as follows: for each pair of Players $i,k$ and $j\in S_i,\ell\in S_k$,
the payoff to Player $i$ in her edge-game with Player $k$ when Player $i$ picks $j$ and Player $k$ picks $\ell$ is
$H^{ik}_{j\ell} := U^{ik}_{j\ell}$, where $U^{ik}_{j\ell}$ is defined in Eqn.~\eqref{eq:U}.
\end{theorem}

This theorem shows that for any game $\bbG$, the value of $C_\bbG(\bbp)$ is the same as in a particular graphical game,
where each pair of players, $(i, k)$ play a bimatrix game whose utility is exactly the utility of the original game $\bbG$,
but taking the expectation on the randomness of the other players' strategies.
If the original game $\bbG$ is a graphical game, then in the graphical game $H^{ik}_{j\ell} = U^{ik}_{j\ell} + c_{-i, -k}$,
where $c_{-i, -k}$ is a parameter which does not depend on Players $i$ and $k$.

Theorem~\ref{thm:equivalence} will be used in Appendix~\ref{app:equivalence} to show the following proposition,
which shows that the volume-changing behaviors of MWU and OMWU are \emph{opposite} to each other in multi-player game,
generalizing a prior result in~\cite{CP2020}.

\begin{proposition}\label{pr:opposite-multi}
The volume integrands of MWU and OMWU in a multi-player game $\bbG$ are respectively $1 + C_\bbG(\bbp) \cdot \ep^2 + \calO(\ep^3)$
and $1 - C_\bbG(\bbp) \cdot \ep^2 + \calO(\ep^3)$. Thus, volume expands locally around a dual point $\bbp$ for MWU (resp.~OMWU) if $C_\bbG(\bbp)$ is positive (resp.~negative).
\end{proposition}

\parabold{Multiplayer Potential Game.}
%
%
By Observation~\ref{obs::potential::bimatrix}, we know that the volume behavior of a potential game is equivalent to a corresponding coordination game in bimatrix game. In this section, we want to show, this holds even in the multi player setting.
\begin{lemma}\label{lem::multi::potential}
Suppose $\mathcal{P}$ is the potential function of a potential game $\bbU$. Let $\bbU^{\mathcal{P}}$ be a game that all players will receive $\mathcal{P}(\bbs)$ when players play strategies $\bbs$. Then
$C_{\bbU}(\bbp) = C_{\bbU^{\mathcal{P}}}(\bbp) \leq 0$.
\end{lemma}

In Appendix~\ref{sec::multi::potential}, we will discuss some situations where $C_{\bbU}(\bbp)$ is strictly less than $0$, thus OMWU is Lyapunov chaotic therein.

\bibliographystyle{plain}
\bibliography{game}

\newpage

\appendix

\normalsize

\section{Proofs in Section~\ref{sec::bimatrix}}

\begin{proof}[Proof of Lemma~\ref{lem::trivial::matrix}]
First, observe that it suffices to prove that the lemma holds when $\bbT^1$ is a trivial matrix and $\bbT^2$ is the zero matrix.
Then the lemma holds for any trivial matrices $\bbT^1,\bbT^2$ due to symmetry:
$C_{(\bbA,\bbB)}(\bbp,\bbq) = C_{(\bbA+\bbT^1,\bbB)}(\bbp,\bbq) = C_{(\bbA+\bbT^1,\bbB+\bbT^2)}(\bbp,\bbq)$.

Due to the definition of trivial matrix, we can write $T^1_{jk} = u_j + v_k$. Then
\begin{align*}
&C_{(\bbA+\bbT^1,\bbB)}(\bbp,\bbq) - C_{(\bbA,\bbB)}(\bbp,\bbq)\\
~=~ & - \expect{\left(A_{jk} + u_j + v_k - A_j - u_j - \sum_{\ell\in S_2} v_\ell y_\ell - A_k - v_k - \sum_{\ell\in S_1} u_\ell x_\ell\right)(B_{jk}-B_j-B_k)}\\
&\quad + \expect{A_{jk}+u_j + v_k} \cdot \expect{B_{jk}} + \expect{(A_{jk}-A_j-A_k)(B_{jk}-B_j-B_k)} - \expect{A_{jk}}\cdot \expect{B_{jk}}\\
~=~ & -\expect{\left( - \sum_{\ell\in S_2} v_\ell y_\ell - \sum_{\ell\in S_1} u_\ell x_\ell \right)(B_{jk}-B_j-B_k)} + \expect{u_j + v_k} \cdot \expect{B_{jk}}\\
~=~ & \expect{v_k + u_j} \cdot \expect{B_{jk}-B_j-B_k} + \expect{u_j + v_k} \cdot \expect{B_{jk}}.
\end{align*}
By recalling that $\expect{B_{jk}-B_j-B_k} = -\expect{B_{jk}}$, we have $C_{(\bbA+\bbT^1,\bbB)}(\bbp,\bbq) - C_{(\bbA,\bbB)}(\bbp,\bbq) = 0$.
\end{proof}

\begin{proof}[Proof of Observation~\ref{obs::potential::bimatrix}]
Let $\mathcal{P}_{jk}$ be the potential value of a potential game when Player $1$ plays strategy $j$ and Player $2$ plays strategy $k$.
Then according to the definition the potential function, for any $j_1$, $j_2$ and $k$,
\begin{align*}
  A_{j_1 k} - A_{j_2 k} = \mathcal{P}_{j_1 k} - \mathcal{P}_{j_2 k}.
  \end{align*}
In particular, for any $j,k$, $A_{jk} = \calP_{jk} + A_{1k} - \calP_{1k}$.
This implies that there exists $\bbv$ such that $A_{jk} = \mathcal{P}_{jk} + v_k$ for any $j$ and $k$.

Similarly, there exists $\bbu$ such that $B_{jk} = \mathcal{P}_{jk} + u_j$ for any $j$ and $k$.
This implies that any two-player potential games are coordination games plus trivial matrices.
\end{proof}

\begin{proof}[Proof of Theorem~\ref{thm:matrix-domination}]
We first prove that if $\bbZ$ dominates $\bbC$, then $C_{(\bbA, \bbB)}(\bbp, \bbq)$ is always non-negative. By Observation~\ref{obs::C::dec},
  \begin{align*}
    C_{(\bbA, \bbB)}(\bbp, \bbq) &= C_{(\bbZ, -\bbZ)}(\bbp, \bbq) + C_{(\bbC, \bbC)}(\bbp, \bbq) \\
    &= C_{(\bbZ, -\bbZ)}(\bbp, \bbq) - C_{(\bbC, -\bbC)}(\bbp, \bbq) \\
    &= \frac{1}{4} \sum_{j, j', k, k'} x_j(\bbp) y_k(\bbq) x_{j'}(\bbp) y_{k'}(\bbq) \cdot \\
    &~~~~~~~~~~~~\left( \left( Z_{jk} + Z_{j'k'} - Z_{jk'} - Z_{j'k}\right)^2 - \left( C_{jk} + C_{j'k'} - C_{jk'} - C_{j'k}\right)^2 \right) \\
    &\geq 0.
  \end{align*}
  In contrast, if $\bbZ$ does not dominate $\bbC$, then there exist $\hat{j}$, $\hat{j}'$, $\hat{k}$, $\hat{k}'$ and $\delta > 0$ such that
  \begin{align*}
    \left( C_{\hat{j}\hat{k}} + C_{\hat{j}'\hat{k}'} - C_{\hat{j}\hat{k}'} - C_{\hat{j}'\hat{k}}\right)^2 \geq \left( Z_{\hat{j}\hat{k}} + Z_{\hat{j}'\hat{k}'} - Z_{\hat{j}\hat{k}'} - Z_{\hat{j}'\hat{k}}\right)^2 + \delta.
  \end{align*}
  For each $\eta > 0$, we construct $\bbp$ and $\bbq$ such that $x_{\hat{j}}(\bbp) = x_{\hat{j}'}(\bbp) = y_{\hat{k}}(\bbq) = y_{\hat{k}'}(\bbq) = \frac{1 - \eta}{2}$.
  Furthermore, we let $\Upsilon$ denote the maximum absolute value of all entries in matrices $\bbA$ and $\bbB$. Then, for all $j$ and $k$, $|Z_{jk}| \leq \Upsilon$ and $|C_{jk}| \leq \Upsilon$.
  Therefore,
  \begin{align*}
    C_{(\bbA, \bbB)}(\bbp, \bbq) &= C_{(\bbZ, -\bbZ)}(\bbp, \bbq) + C_{(\bbC, \bbC)}(\bbp, \bbq) \\
    &= C_{(\bbZ, -\bbZ)}(\bbp, \bbq) - C_{(\bbC, -\bbC)}(\bbp, \bbq) \\
    &= \frac{1}{4} \sum_{j, j', k, k'} x_j(\bbp) y_k(\bbq) x_{j'}(\bbp) y_{k'}(\bbq) \cdot \\
    &~~~~~~~~~~~~\left( \left( Z_{jk} + Z_{j'k'} - Z_{jk'} - Z_{j'k}\right)^2 - \left( C_{jk} + C_{j'k'} - C_{jk'} - C_{j'k}\right)^2 \right) \\
    &\leq - \delta \left(\frac{1 - \eta}{2}\right)^4 + |S_1|^2 \cdot |S_2|^2 \cdot \eta \cdot 16 \Upsilon^2.
  \end{align*}
  The last inequality holds as $\left( C_{jk} + C_{j'k'} - C_{jk'} - C_{j'k}\right)^2 - \left( Z_{jk} + Z_{j'k'} - Z_{jk'} - Z_{j'k}\right)^2 \leq 16 \Upsilon^2$. The value of $- \delta \left(\frac{1 - \eta}{2}\right)^4 + |S_1|^2 \cdot |S_2|^2 \cdot \eta \cdot 16 \Upsilon^2$ will be negative if we pick a small enough $\eta$.
\end{proof}

\begin{proof}[Proof of Observation~\ref{obs::C::dec}]
Consider a random process, where $j,j'\in S_1$ are randomly picked according to distribution $\bbx(\bbp)$, and $k,k'\in S_2$ are randomly picked according to distribution $\bby(\bbq)$.
Then the RHS of Observation~\ref{obs::C::dec} can be expressed as $\frac 14 \cdot \expect{(Z_{jk} + Z_{j'k'} - Z_{jk'} - Z_{j'k})^2}$.

Then we expand the squared term in the expectation. Observing the symmetries within the expansion, we immediately have
\[
\frac 14 \cdot \expect{(Z_{jk} + Z_{j'k'} - Z_{jk'} - Z_{j'k})^2} = \expect{(Z_{jk})^2} - \expect{Z_{jk} Z_{jk'}} - \expect{Z_{jk} Z_{j'k}} + \expect{Z_{jk} Z_{j'k'}}.
\]
Let $Z_j := [\bbZ \bby]_j$ and $Z_k = [\bbZ\trans \bbx]_k$. Then we have
\[
\expect{Z_{jk} Z_{jk'}} = \sum_{j,k} x_j y_k Z_{jk} \sum_{k'} y_{k'} Z_{jk'} = \sum_{j} x_j [\bbZ \bby]_j \sum_k y_k Z_{jk} = \sum_j x_j (Z_j)^2 = \expect{(Z_j)^2}.
\]
Similarly, $\expect{Z_{jk} Z_{j'k}} = \expect{(Z_k)^2}$. Lastly, $\expect{Z_{jk} Z_{j'k'}} = \expect{Z_{jk}}^2$.
Thus, the RHS of Observation~\ref{obs::C::dec} is simplified to
\[
\expect{(Z_{jk})^2} - \expect{(Z_j)^2} - \expect{(Z_k)^2} + \expect{Z_{jk}}^2.
\]

We complete the proof by noting that from the definition of $C_{(\bbZ,-\bbZ)}(\cdot)$ in Eqn.~\eqref{eq:Cxy},
$C_{(\bbZ,-\bbZ)}(\cdot)$ can be rewritten as
\[
\expect{(Z_{jk})^2} - \expect{Z_j Z_{jk}} - \expect{Z_k Z_{jk}} + \expect{Z_{jk}}^2,
\]
while $\expect{Z_j Z_{jk}} = \sum_j x_j Z_j \sum_k y_k Z_{jk} = \sum_j x_j Z_j Z_j = \expect{(Z_j)^2}$,
and similarly $\expect{Z_j Z_{jk}} = \expect{(Z_k)^2}$.
\end{proof}

\begin{proof}[Proof of Theorem~\ref{thm:matrix-strict-domination}]
We only need to prove $\bar{c}(S) = \inf_{(\bbp,\bbq)\in S} C_{(\bbA,\bbB)}(\bbp,\bbq) \geq \delta^2 \theta^2$. This is because matrix $\bbZ$ $\theta$-dominates $\bbC$, which implies there exist $j$, $j'$, $k$, and $k'$ such that
\begin{align*}
  \left(\bbZ_{jk} + \bbZ_{j'k'} - \bbZ_{jk'} - \bbZ_{j'k}\right)^2 \geq \left(\bbQ_{jk} + \bbQ_{j'k'} - \bbQ_{jk'} - \bbQ_{j'k}\right)^2 + \theta^2.
\end{align*}
By applying Observation~\ref{obs::C::dec}, $C_{(\bbZ, -\bbZ)}(\bbp, \bbq)  \geq C_{(\bbC, -\bbC)}(\bbp, \bbq) + \theta^2 \delta^2$, because
every $x_j(\bbp),y_k(\bbq)$ for $(\bbp,\bbq)\in S$ is at least $\delta$.
By noting that $C_{(\bbC, \bbC)}(\bbp, \bbq) = -C_{(\bbC, -\bbC)}(\bbp, \bbq)$ and $C_{(\bbA, \bbB)}(\bbp, \bbq) = C_{(\bbZ, -\bbZ)}(\bbp, \bbq) + C_{(\bbC, \bbC)}(\bbp, \bbq)$, the result follows.
\end{proof}

\begin{proof}[Proof of Lemma~\ref{lem::linear::program}]
A key observation is
\[
C_{(\bbC, -\bbC)}(\bbp, \bbq) = \min_{\bbg, \bbh}~\sum_{jk} x_j(\bbp) y_k(\bbq) \left(C_{jk} - g_j - h_k\right)^2.
\]
With this observation and comparing this with the definition of $r(\bbZ)$, it's easy to figure out that $C_{(\bbC, -\bbC)}(\bbp, \bbq) \leq r(\bbC)^2$.

To see $C_{(\bbC, -\bbC)}(\bbp, \bbq) \geq (r(\bbC)) \cdot \delta)^2$, we first let $\bbg^*$ and $\bbh^*$ to be the optimal choice of $\bbg$ and $\bbh$ in $C_{(\bbC, -\bbC)}(\bbp, \bbq) = \min_{\bbg, \bbh} \sum_{j,k} x_j(\bbp) y_k(\bbq) \left(C_{jk} - g_j - h_k\right)^2$. One immediate observation is \footnote{If this is not true, we can let $\bbg$ and $\bbh$ in $r(\bbZ)$ to be $\bbg_j = \bbg^*_j - \frac{\max_{j, k} \{C_{jk} - g^*_j - h^*_k\} +  \min_{j, k} \{C_{jk} - g^*_j - h^*_k\} }{2}$ and $\bbh_k = \bbh^*_k$. Then we can achieve $r(\bbC) =  \frac{\max_{j, k} \{C_{jk} - g^*_j - h^*_k\} -  \min_{j, k} \{C_{jk} - g^*_j - h^*_k\} }{2}$ which make $r(\bbC)$ smaller.
}
\begin{align*}
2 r(\bbC) \leq \max_{j, k} \{C_{jk} - g^*_j - h^*_k\} -  \min_{j, k} \{C_{jk} - g^*_j - h^*_k\}.
\end{align*}
Therefore,
\begin{align*}
  \max\left\{ \left(\max_{j, k} \{C_{jk} - g^*_j - h^*_k\}\right)^2, \left(\min_{j, k} \{C_{jk} - g^*_j - h^*_k\}\right)^2 \right\} \geq r(\bbC)^2.
\end{align*}
This immediately implies that $C_{(\bbC, -\bbC)}(\bbp, \bbq) \geq (r(\bbC)) \cdot \delta)^2$.
\end{proof}

\section{Local Equivalence of Volume Change between Normal-form and Graphical Games}\label{app:equivalence}

In this appendix, we concern the volume change of a learning algorithm in multi-player game.
We first recap from~\cite{CP2020} on how the volume change is computed for dynamical systems which are \emph{gradual} (i.e.~those governed by a small step-size),
followed by a continuous-time analogue of OMWU in games, which are crucial for analyzing the volume change of discrete-time OMWU.
Then we compute the volume changes of MWU and OMWU in multi-player graphical games and normal-form games respectively.
Once these are done, the proofs of Proposition~\ref{pr:opposite-multi} and Theorem~\ref{thm::linear::program} become apparent.

\subsection{Discrete-Time Dynamical Systems and Volume of Flow}

\newcommand{\vol}{\mathsf{vol}}

We consider discrete-time dynamical systems in $\rr^d$. Such a dynamical system is determined recursively by a starting point $\bbs(0)\in \rr^d$
and an update rule of the form $\bbs(t+1) = G(\bbs(t))$, for some function $G:\rr^d \ra \rr^d$.
Here, we focus on the special case when the update rule is \emph{gradual}, i.e.~it is in the form of
\[
\bbs(t+1) = \bbs(t) + \ep \cdot F(\bbs(t)),
\]
where $F:\rr^d\ra \rr^d$ is a smooth function and step-size $\ep>0$.
When $F$ and $\ep$ are given, the flow of the starting point $\bbs(0)$ at time $t$, denoted by $\Phi(t,\bbs(0))$, is simply the point $\bbs(t)$ generated by the above recursive update rule.
Then the flow of a set $S\subset \rr^d$ at time $t$, denoted by $\Phi(t,S)$, is the set $\left\{ \Phi(t,\bbs) ~|~ \bbs\in S \right\}$.
Since $F$ does not depend on time $t$, we have the following equality: $\Phi(t_1+t_2,S) = \Phi(t_2,\Phi(t_1,S))$.

By equipping $\rr^d$ with the standard Lebesgue measure, the \emph{volume} of a measurable set $S$, denoted by $\vol(S)$, is simply its measure.
Given a bounded and measurable set $S\subset \rr^d$, if the discrete flow in one time step maps $S$ to $S' = \Phi(1,S)$ injectively, then
by integration by substitution for multi-variables,
\begin{equation}\label{eq:Liouville-discrete}
\vol(S') ~=~ \int_{\bbs\in S} \det \left( \bbI + \ep \cdot \bbJ(\bbs) \right) \,\mathsf{d}V,
\end{equation}
where $\bbI$ is the identity matrix, and $\bbJ(\bbs)$ is the \emph{Jacobian} matrix defined below:
\begin{equation}\label{eq:Jacobian}
\bbJ(\bbs) ~=~
\begin{bmatrix}
\frac{\partial}{\partial s_1}F_1(\bbs) & \frac{\partial}{\partial s_2}F_1(\bbs) & \cdots & \frac{\partial}{\partial s_d}F_1(\bbs)\\
\frac{\partial}{\partial s_1}F_2(\bbs) & \frac{\partial}{\partial s_2}F_2(\bbs) & \cdots & \frac{\partial}{\partial s_d}F_2(\bbs)\\
\vdots & \vdots & \ddots & \vdots \\
\frac{\partial}{\partial s_1}F_d(\bbs) & \frac{\partial}{\partial s_2}F_d(\bbs) & \cdots & \frac{\partial}{\partial s_d}F_d(\bbs)
\end{bmatrix}.
\end{equation}

Clearly, analyzing the determinant in the integrand in~\eqref{eq:Liouville-discrete} is crucial in volume analysis; we call it the \emph{volume integrand}.
When the determinant is expanded using the Leibniz formula, it becomes a polynomial of $\ep$, in the form of $1 + C(\bbs) \cdot \ep^h + \calO(\ep^{h+1})$ for some integer $h\ge 1$.
Thus, when the step-size $\ep$ is sufficiently small, the sign of $C(\bbs)$ dictates on whether the volume expands or contracts.

\subsection{Continuous-Time Analogue of OMWU}

OMWU does not fall into the category of dynamical systems defined above, since its update rule is in the form of $\bbs(t+1) = G(\bbs(t),\bbs(t-1))$.
Fortunately, Cheung and Piliouras~\cite{CP2020} showed that OMWU can be well-approximated by the \emph{online Euler discretization} of a system of ordinary differential equations (ODE),
and thus it can be well-approximated by a dynamical system.

The ODE system is given below. $\bbp$ is a dual (cumulative payoff) vector variable,
$\bbu:\rrplus\ra \rr^d$ is the function such that $\bbu(t)$ gives the \emph{instantaneous payoff} vector at time $t$.
We assume that $\bbu$ is twice differentiable with bounded second-derivatives, and $\dot{\bbu}$ denotes the time-derivative of $\bbu$.
\begin{equation}\label{eq:DE-OptMD-general}
\dot \bbp ~=~ \bbu ~+~ \ep \cdot \dot{\bbu},
\end{equation}

\emph{Online Euler discretization} (OED) of~\eqref{eq:DE-OptMD-general} refers to the following time-discretization of the ODE system.
In applications, $\dot{\bbu}$ might not be explicitly given, and the sequence $\bbu(0),\bbu(1),\bbu(2),\cdots$
are available \emph{online} (i.e., at time $t$ we only have access of $\bbu(\tau)$ for $\tau=0,1,\cdots,t$).
As the discretization step is $\ep$, we approximate $\dot \bbu(t)$ by $(\bbu(t) - \bbu(t-1)) / \ep$.
By using this approximation, OED of~\eqref{eq:DE-OptMD-general} yields
\[
\bbp(t+1) ~=~ \bbp(t) ~+~ \ep \cdot \left[ \bbu(t) ~+~ \ep \cdot \frac{\bbu(t) - \bbu(t-1)}{\ep} \right] ~=~ \bbp(t) ~+~ \ep \cdot \left[ 2\cdot \bbu(t) - \bbu(t-1) \right],
\]
which is exactly the OMWU update rule in general context.

When compared the OED with the standard Euler discretization
\[
\bbp(t+1) ~=~ \bbp(t) ~+~ \ep \cdot \left[ \bbu(t) ~+~ \ep \cdot \dot{\bbu}(t) \right],
\]
OED incurs a local error that appears due to the approximation of $\dot \bbu(t)$. The local error can be bounded by $\calO(\ep^3)$.
Cheung and Piliouras~\cite{CP2020} showed that eventually the determinant of the volume integrand is a of the form $1 + C(\bbs) \cdot \ep^2 + \calO(\ep^3)$,
the local error does not affect the first and second highest-order terms, and hence can be ignored henceforth.

\subsection{MWU in Graphical Games}

Let $\bbH$ be a graphical game of $N$ players, where between every pair of Players $i$ and $k$, the payoff bimatrices are $(\bbH^{ik},(\bbH^{ki})\trans)$.
In the dual space, let $\bbp = (\bbp_1,\cdots,\bbp_N)$ denote the cumulative payoff profile,
and let $\bbx = (\bbx_1,\cdots,\bbx_N)$ denote the corresponding mixed strategy profile,
where $\bbx_i$ is a function of $\bbp_i$.
We will write $\bbx_i$ and $\bbx_i(\bbp_i)$ interchangeably.
The expected payoff to strategy $j$ of Player $i$ is
\[
u_{ij}(\bbp) ~=~ \sum_{\substack{k\in [N]\\k\neq i}} [\bbH^{ik} \cdot \bbx_k(\bbp_k)]_j,
\]
which will be used to compute the Jacobian matrices of MWU and OMWU.

For MWU, the Jacobian matrix $\bbJ$ is a squared matrix with each row and each column indexed by $(i,j)$, where $i$ is a Player and $j\in S_i$.
The precise values of its entries are given below:
\begin{equation}\label{eq:MWU-Jacob-1}
\forall j_1,j_2\in S_i,~~\ep J_{(i,j_1),(i,j_2)} ~=~ \ep \cdot \frac{\partial u_{ij_1}}{\partial p_{ij_2}} ~=~ 0
\end{equation}
and
\begin{equation}\label{eq:MWU-Jacob-2}
\forall i\neq k,~j\in S_i,~\ell\in S_k,~~\ep J_{(i,j),(k,\ell)} ~=~ \ep \cdot \frac{\partial u_{ij}}{p_{k\ell}} ~=~ \ep x_{k\ell} \cdot \left( H^{ik}_{j\ell} - [\bbH^{ik} \cdot \bbx_k]_j \right).
\end{equation}
Then by expansion using Leibniz formula, the determinant of $(\bbI + \ep \cdot \bbJ)$ is
\begin{align}
&~1 - \sum_{\substack{i\in [N]\\j\in S_i}} ~\sum_{\substack{k > i\\\ell\in S_k}} ~ (\ep J_{(i,j),(k,\ell)}) (\ep J_{(k,\ell),(i,j)}) ~+~ \calO(\ep^3) \nonumber\\
=&~1 ~-~ \ep^2 \cdot \sum_{\substack{i\in [N]\\j\in S_i}} ~\sum_{\substack{k > i\\\ell\in S_k}} ~x_{ij} x_{k\ell}
\left( H^{ki}_{\ell j} - [\bbH^{ki} \cdot \bbx_i]_\ell \right) \left( H^{ik}_{j\ell} - [\bbH^{ik} \cdot \bbx_k]_j \right) ~+~ \calO(\ep^3).\label{eq:MWU-graph-1}
\end{align}
By noting the similarity of the double summation to $C_{(\bbA,\bbB)}(\cdot)$ in~\eqref{eq:Cxy}, we can immediately rewrite the above expression as
\begin{equation}\label{eq:vol-int-MWU-graphical}
1 ~+~ \ep^2 \cdot \sum_{i,k: 1\le i<k\le N} C_{(\bbH^{ik},(\bbH^{ki})\trans)} (\bbp_i,\bbp_k) ~+~ \calO(\ep^3).
\end{equation}

\subsection{OMWU in Graphical Games}

For OMWU, as we pointed out already, we will first consider its continuous analogue first. Thus, we need to compute $\dot{\bbu}$ in the continuous-time setting.
By chain rule, we have
\[
\dot{u}_{ij}(\bbp) ~=~ \sum_{\substack{k\in [N]\\k\neq i\\ \ell\in S_k}} \frac{\partial{[\bbH^{ik} \cdot \bbx_k(\bbp_k)]_j}}{\partial p_{k\ell}} \cdot \difft{p_{k \ell}}
~=~ \sum_{\substack{k\in [N]\\k\neq i\\ \ell\in S_k}} x_{k\ell} \cdot \left( H^{ik}_{j\ell} - [\bbH^{ik} \cdot \bbx_k]_j \right) \cdot \difft{p_{k \ell}},
\]
and hence
\[
\difft{p_{ij}} ~=~ \sum_{\substack{k\in [N]\\k\neq i}} [\bbH^{ik} \cdot \bbx_k]_j ~+~ \ep \cdot \sum_{\substack{k\in [N]\\k\neq i\\ \ell\in S_k}}
x_{k\ell} \cdot \left( H^{ik}_{j\ell} - [\bbH^{ik} \cdot \bbx_k]_j \right) \cdot \difft{p_{k \ell}}.
\]
Note that this is a recurrence formulae for $\difft{\bbp}$. By iterating it\footnote{For the formality on why we can do iterations when $\ep$ is sufficiently small, see~\cite{CP2020}.}, we have
\[
\difft{p_{ij}} ~=~ \sum_{\substack{k\in [N]\\k\neq i}} [\bbH^{ik} \cdot \bbx_k]_j ~+~ \ep \cdot \sum_{\substack{k\in [N]\\k\neq i\\ \ell\in S_k}}
x_{k\ell} \cdot \left( H^{ik}_{j\ell} - [\bbH^{ik} \cdot \bbx_k]_j \right) \cdot \left( \sum_{\substack{r\in [N]\\r\neq k}} [\bbH^{kr} \cdot \bbx_r]_\ell \right) ~+~ \calO(\ep^2).
\]
Hence, its standard Euler discretization, which approximates the OED with local error $\calO(\ep^3)$, can be written as below (where we ignore the $\calO(\ep^3)$ error terms):
\[
p_{ij}(t+1) ~=~ p_{ij}(t) ~+~ \ep \sum_{\substack{k\in [N]\\k\neq i}} [\bbH^{ik} \cdot \bbx_k]_j ~+~ \ep^2 \sum_{\substack{k\in [N]\\k\neq i\\ \ell\in S_k}}
x_{k\ell} \cdot \left( H^{ik}_{j\ell} - [\bbH^{ik} \cdot \bbx_k]_j \right) \cdot \left( \sum_{\substack{r\in [N]\\r\neq k}} [\bbH^{kr} \cdot \bbx_r]_\ell \right).
\]
With this, we are ready to compute the Jacobian matrix $\bbJ$ for OMWU. For all $j_1,j_2\in S_i$,
\begin{equation}\label{eq:OMWU-Jacob-1}
\ep J_{(i,j_1),(i,j_2)} ~=~ \ep^2 \sum_{\substack{k\in [N]\\k\neq i\\ \ell\in S_k}} x_{k\ell} \cdot \left( H^{ik}_{j_1\ell} - [\bbH^{ik} \cdot \bbx_k]_{j_1} \right)
\cdot x_{ij_2}\cdot \left( H^{ki}_{\ell j_2} - [\bbH^{ki} \cdot \bbx_i]_\ell \right)
\end{equation}
and for all $i\neq k$, $j\in S_i$, $\ell\in S_k$,
\begin{equation}\label{eq:OMWU-Jacob-2}
\ep J_{(i,j),(k,\ell)} ~=~ \ep x_{k\ell} \left( H^{ik}_{j\ell} - [\bbH^{ik} \cdot \bbx_k]_j\right) + \calO(\ep^2)
\end{equation}
Then by expansion using Leibniz formula, the determinant of $(\bbI + \ep \cdot \bbJ)$ is
\[
1 + \left( \underbrace{\sum_{\substack{i\in [N]\\j\in S_i}} \ep J_{(i,j),(i,j)}}_{T_1} ~-~ \underbrace{\sum_{\substack{i\in [N]\\j\in S_i}}~\sum_{\substack{k>i\\ \ell\in S_k}}~(\ep J_{(i,j),(k,\ell)}) (\ep J_{(k,\ell),(i,j)})}_{T_2}  \right) ~+~ \calO(\ep^3).
\]
By a direct expansions on $T_1$ and $T_2$, it is easy to see that $T_1 = 2 T_2$ (after ignoring $\calO(\ep^3)$ terms). 
On the other hand, the coefficient of $\ep^2$ in $T_2$ is exactly the same as the double summation in~\eqref{eq:MWU-graph-1},
thus it equals to $-\sum_{i,k: 1\le i<k\le N} C_{(\bbH^{ik},(\bbH^{ki})\trans)} (\bbp_i,\bbp_k)$.
Overall, we show that the determinant equals to
\begin{equation}\label{eq:vol-int-OMWU-graphical}
1 ~-~ \ep^2 \cdot \sum_{i,k: 1\le i<k\le N} C_{(\bbH^{ik},(\bbH^{ki})\trans)} (\bbp_i,\bbp_k) ~+~ \calO(\ep^3).
\end{equation}

\begin{obs}\label{obs:oppo-graphical}
The coefficient of $\ep^2$ in~\eqref{eq:vol-int-OMWU-graphical} is the exact negation of the coefficient of $\ep^2$ in~\eqref{eq:vol-int-MWU-graphical}.
\end{obs}

\subsection{Completing the Local Equivalence Proof}

In a multiplayer normal-form game $\bbG$, recall that notation~\eqref{eq:U}.
We point out the following formulae:
\begin{align*}
\frac{\partial U^{i_1 i_2 \cdots i_g}_{j_1 j_2 \cdots j_g}}{\partial p_{ij}} &~=~ 0\hspace*{2in}\text{if}~i\in \{i_1,i_2,\cdots,i_g\};\\
\frac{\partial U^{i_1 i_2 \cdots i_g}_{j_1 j_2 \cdots j_g}}{\partial p_{ij}} &~=~ x_{ij} \cdot \left( U^{i_1 i_2 \cdots i_g i}_{j_1 j_2 \cdots j_g j} - U^{i_1 i_2 \cdots i_g}_{j_1 j_2 \cdots j_g} \right)\hspace*{0.4in}\text{if}~i\notin \{i_1,i_2,\cdots,i_g\}.
\end{align*}

\paragraph{MWU.}
Here, MWU update rule is $p_{ij}(t+1) = p_{ij}(t) + \ep \cdot U^i_j$. When computing the Jacobian matrix for this update rule using the formulae above,
and comparing it with the Jacobian matrix computed in~\eqref{eq:MWU-Jacob-1} and~\eqref{eq:MWU-Jacob-2},
it is immediate that they are the same by setting $H^{ik}_{j\ell} = U^{ik}_{j\ell}$.
This derives~\eqref{eq:Cxy-multi}, and completes the proof of Theorem~\ref{thm:equivalence}.

\paragraph{OMWU.}
As before, we use the continuous analogue and compute $\dot{\bbu}$. By the chain rule and the above formulae, we have
\[
\dot{u}_{ij}(\bbp) ~=~ \sum_{\substack{k\in [N]\\k\neq i\\ \ell\in S_k}} \frac{\partial U^i_j}{\partial p_{k\ell}}\cdot \difft{p_{k\ell}} ~=~
\sum_{\substack{k\in [N]\\k\neq i\\ \ell\in S_k}} x_{k\ell} \cdot \left( U^{ik}_{j\ell} - U^i_j \right) \cdot \difft{p_{k\ell}}
\]
and hence
\[
\difft{p_{ij}} ~=~ U^i_j ~+~ \ep \cdot \sum_{\substack{k\in [N]\\k\neq i\\ \ell\in S_k}} x_{k\ell} \cdot \left( U^{ik}_{j\ell} - U^i_j \right) \cdot \difft{p_{k\ell}}.
\]
Iterating the above recurrence yields
\[
\difft{p_{ij}} ~=~ U^i_j ~+~ \ep \cdot \sum_{\substack{k\in [N]\\k\neq i\\ \ell\in S_k}} x_{k\ell} \cdot \left( U^{ik}_{j\ell} - U^i_j \right) \cdot U^k_\ell ~+~ \calO(\ep^2).
\]
Its standard Euler discretization is
\[
p_{ij}(t+1) ~=~ p_{ij}(t) + \ep \cdot U^i_j ~+~ \ep^2 \cdot \sum_{\substack{k\in [N]\\k\neq i\\ \ell\in S_k}} x_{k\ell} \cdot \left( U^{ik}_{j\ell} - U^i_j \right) \cdot U^k_\ell.
\]
Now we compute the Jacobian matrix for this standard Euler discretization. For $j_1,j_2\in S_i$,
\[
\ep J_{(i,j_1),(i,j_2)} ~=~ \ep^2 \sum_{\substack{k\in [N]\\k\neq i\\ \ell\in S_k}} x_{k\ell} \cdot \left( U^{ik}_{j_1\ell} - U^i_{j_1} \right) \cdot x_{ij_2} \cdot \left(U^{ki}_{\ell j_2} - U^k_\ell \right)
\]
and for all $i\neq k$, $j\in S_i$, $\ell\in S_k$,
\[
\ep J_{(i,j),(k,\ell)} ~=~ \ep x_{k\ell} \left(U^{ik}_{j\ell} - U^i_j\right) + \calO(\ep^2).
\]
By comparing this computed Jacobian matrix with the Jacobian matrix computed in~\eqref{eq:OMWU-Jacob-1} and~\eqref{eq:OMWU-Jacob-2},
it is immediate to see that their determinants are the same (after ignoring all $\calO(\ep^3)$ terms) by setting $H^{ik}_{j\ell} = U^{ik}_{j\ell}$.
With the result we just derived, together with Observation~\ref{obs:oppo-graphical} and Theorem~\ref{thm:equivalence}, Proposition~\ref{pr:opposite-multi} follows.

\section{Multi-player Potential Game}\label{sec::multi::potential}
\begin{proof}[Proof of Lemma~\ref{lem::multi::potential}]
  We know that the potential game satisfies the following condition:
  \begin{align*}
    \mathcal{P}(s_i, s_{-i}) - \mathcal{P}(s'_i, s_{-i}) = u_i(s_i, s_{-i}) - u_i(s_{i'}, s_{-i}).
  \end{align*}
  Therefore, $u_i(s_i, s_{-i}) = \mathcal{P}(s_i, s_{-i}) + v^i(s_{-i})$. Note that $v^i(s_{-i})$ does not depend on $s_i$, the strategy of player $i$.

  By Theorem~\ref{thm:equivalence}, let $\bbH(\bbU)$ be the induced graphical game of $\bbU$ and $\bbH(\bbU^{\mathcal{P}})$ be the induced graphical game of $\bbU^{\mathcal{P}}$. Then,
  \begin{align*}
    C_{\bbU}(\bbp) &= C_{\bbH(\bbU)}(\bbp)& \mbox{(Theorem~\ref{thm:equivalence})}\\
    & = \sum_{i, k} C_{(\bbH(\bbU)^{ik}, (\bbH(\bbU)^{ki})\trans)}(\bbp_i, \bbp_k) &\mbox{(By \eqref{eq:vol-int-MWU-graphical})} \\
    & = \sum_{i, k} C_{(\bbH(\bbU^{\mathcal{P}})^{ik}, (\bbH(\bbU^{\mathcal{P}})^{ki})\trans)}(\bbp_i, \bbp_k) & \mbox{(see explanation below)}\\
    & = C_{\bbH(\bbU^{\mathcal{P}})}(\bbp) & \mbox{(By \eqref{eq:vol-int-MWU-graphical})} \\ 
    & = C_{\bbU^{\mathcal{P}}}(\bbp).&\mbox{(Theorem~\ref{thm:equivalence})}
    \end{align*}
    The third equality holds as the difference between $\bbH(\bbU)^{ik}$ and $\bbH(\bbU^{\mathcal{P}})^{ik}$ is a trivial matrix:
    \begin{align*}
      \bbH(\bbU)^{ik}_{jl} = \bbU^{ik}_{jl} = \left(\bbU^{\mathcal{P}}\right)^{ik}_{jl} + \mathbf{E}_{-(i,k)}\left[v^i(s_{-i})\right] = \bbH(\bbU^{\mathcal{P}})^{ik}_{jl} + \mathbf{E}_{-(i,k)}\left[v^i(s_{-i})\right];
    \end{align*}
    where \footnote{$\mathbf{E}_{-(i,k)}\left[v^i(s_{-i})\right]$ is the expectation over all the strategies taken by the players other than $i$ and $k$ and $v^i(s_{-i})$ does not depend on the strategy taken by player $i$.} $\mathbf{E}_{-(i,k)}\left[v^i(s_{-i})\right]$ doesn't depend on $j$, the strategy of player $i$, and only depends on $l$, the strategy of player $k$. The same argument applies for $(\bbH(\bbU)^{ki})\trans$ and $(\bbH(\bbU^{\mathcal{P}})^{ki})\trans$.

    To see $C_{\bbU}(\bbp) \leq 0$, observe that the induced graphical game of $\bbU^{\mathcal{P}}$ between player $i$ and $k$, $(\bbH(\bbU^{\mathcal{P}})^{ik}, (\bbH(\bbU^{\mathcal{P}})^{ki})\trans)$, is also a bimatrix coordination game, which implies $C_{(\bbH(\bbU^{\mathcal{P}})^{ik}, (\bbH(\bbU^{\mathcal{P}})^{ki})\trans)}(\cdot) \leq 0$.
    As $C_{\bbU}(\bbp) = \sum_{i, k} C_{(\bbH(\bbU^{\mathcal{P}})^{ik}, (\bbH(\bbU^{\mathcal{P}})^{ki})\trans)}(\bbp_i, \bbp_k)$, the result follows.
\end{proof}

Next, we identify several cases such that $C_{\bbU}(\bbp)$ is strictly negative in the region
\begin{align*}
  S^{\delta} = \{ \bbx | \forall i, j ~ x_{ij} > \delta\}.
\end{align*}
The conditions we pose are on the corresponding potential function $\mathcal{P}$. Note that $\bbH(\bbU^{\mathcal{P}})^{ik}$, the induced edge-game between player $i$ and $k$, is also a coordination game,~i.e. $\bbH(\bbU^{\mathcal{P}})^{ik} = (\bbH(\bbU^{\mathcal{P}})^{ki})\trans$.
\begin{itemize}
  \item Case $1$:
  \begin{align*}
    \min_{\bbx, \bbg, \bbh}\sum_{1\le i < k \le N} ~~\sum_{j\in S_i,\ell\in S_k} ~\left(P^{ik}_{j\ell} - g^{ik}_j - h^{ik}_\ell\right)^2 \ge \theta,
  \end{align*}
   where $P^{ik}_{j\ell} = \mathbb{E}_{\bbs_{-(i, k)}}\left[\mathcal{P}(s_i = j, s_k = \ell, \bbs_{-(i, k)})\right]$. With this condition, we can prove that $C_{\bbU}(\bbp) \leq - \theta \delta^2$ for any $\bbp$ in $S^{\delta}$. One key observation for this is true is that \begin{align*}
   C_{\bbU}(\bbp) &= \sum_{i,k} C_{(\bbH(\bbU^{\mathcal{P}})^{ik}, (\bbH(\bbU^{\mathcal{P}})^{ki})\trans)}(\bbp_i, \bbp_k) \\
   &= - \sum_{i,k} \sum_{j,\ell} x_{ij}(\bbp_i) x_{k\ell}(\bbp_k)\left(P^{ik}_{j\ell} - g^{ik}_j - h^{j\ell}_\ell\right)^2,
 \end{align*}
 as $\bbH(\bbU^{\mathcal{P}})^{ik} = {\bbU^{\mathcal{P}}}^{ik} = \bbP^{ik}$.
 
  \item Case $2$:
  
  If $\bbU$ is a graphical game, then if there exists a pair of player $i_1$ and $i_2$, such that the game between $i_1$ and $i_2$ is a non-trivial game, then $C_{\bbU}$ will be strictly negative in $S^{\delta}$. 

  \item Case $3$:

  Consider the payoff matrix of $\bbU^{\mathcal{P}}$, the coordination game, between players $i_1$ and $i_2$ given a strategy profile of the other players.
  There are total $\prod_{i \neq i_1, i_2} n_i$ such matrices, one for each strategy profile of the other players, and each matrix is of dimension $n_{i_1} \times n_{i_2}$.
  We call these matrices the \emph{projected matrices} for players $i_1,i_2$.

  Let $\calM$ denote the matrix space of $n_{i_1} \times n_{i_2}$. 
  On the other hand, trivial matrices form a subspace of dimension $n_{i_1} + n_{i_2}-1$.\footnote{Recall that
  a trivial matrix $\bbT$ can be represented as $\{u_j + v_k\}_{j, k}$.
  Consider the natural linear map $L$ such that $L(u_1,u_2,\cdots,u_{n_{i_1}},v_1,v_2,\cdots,v_{n_{i_2}})$ maps to the trivial matrix $\bbT$.
  Note that the kernel of $L$ is of dimension $1$, since if $L(u_1,u_2,\cdots,u_{n_{i_1}},v_1,v_2,\cdots,v_{n_{i_2}})$ is the zero matrix,
  then we must have $v_k = -u_j$ for all $j,k$, and hence the kernel of $L$ must be the span of the vector $(\underbrace{1,1,\cdots,1}_{\text{the $u$ part}},\underbrace{-1,-1,\cdots,-1}_{\text{the $v$ part}})$.
  Thus, the dimension of all trivial matrices is the dimension of the domain of $L$, which is $n_{i_1} + n_{i_2}$, minus the dimension of the kernel of $L$.}
  Let's call this the \emph{trivial space}, denoted by $\calT$.

  We consider the direct decomposition $\calM = \calT \oplus \calV$.
  Let a set of bases of $\calM$ be $\bbB_1, \bbB_2, \bbB_3, \cdots, \bbB_{n_{i_1}n_{i_2}}$,
  where the first $n_{i_1} + n_{i_2} - 1$ bases form a basis of $\calT$, and the remaining bases form a basis of $\calV$. Without loss of generality, we assume that all bases are of $\mathtt{L}_2$ norm $1$.\footnote{Here, the norm is defined w.r.t.~the standard Frobenius matrix inner product.}

  Given the above-mentioned bases of $\calM$, each of the projected matrices can be written into a unique linear combination of these bases.
  Now, suppose there is a base $\bbB_\ell$ for $l \ge n_{i_1} + n_{i_2}$ (i.e.~this base is in the set of bases for $\calV$),
  such that all projected matrices have non-positive (or non-negative) coefficients of this base,
  and at least one of these projected matrices (which we call a \emph{special projected matrix}) has strictly negative (or strictly positive) coefficient of the base.
  Then we claim that $C_{(\bbH(\bbU^{\mathcal{P}})^{i_1 i_2}, (\bbH(\bbU^{\mathcal{P}})^{i_2 i_1})\trans)}(\bbp_{i_1}, \bbp_{i_2})$ will be strictly negative in $S^{\delta}$.
  This is because 
  $\bbH(\bbU^{\mathcal{P}})^{i_1 i_2}$ is a convex combination of all those projected matrices,
  and by our assumption above, when $\bbH(\bbU^{\mathcal{P}})^{i_1 i_2}$ is expressed as the linear combinations of the bases of $\calM$,
  the coefficient of $\bbB_\ell$ is strictly negative (or strictly positive),
  thus $\bbH(\bbU^{\mathcal{P}})^{i_1 i_2}$ cannot be a trivial matrix.

  Suppose further that there exists $\theta > 0$ such that a special projected matrix has negative (or positive) coefficient for $\bbB_\ell$ which is smaller (or bigger) than $-\theta$ (or $\theta$),
  then we are guaranteed that $\bbH(\bbU^{\mathcal{P}})^{i_1 i_2}$ is bounded away from $\calT$ for a distance of $\theta \delta^{N-2}$,\footnote{To see why,
  when the coefficient for $\bbB_\ell$ is bounded away from zero, we are guaranteed that the special projected matrix has a strictly positive distance from $\calT$,
  and this distance is at least $\theta$. Then $\bbH(\bbU^{\mathcal{P}})^{i_1 i_2}$, which is a convex combination of all projected matrices
  where each projected matrix (in particular, the special projected matrix) has a weight at least $\delta^{N-2}$,
  has a strictly positive distance from $\calT$ too, which is at least $\theta \delta^{N-2}$.} and hence as the calculations below show,
  $C_{(\bbH(\bbU^{\mathcal{P}})^{i_1 i_2}, (\bbH(\bbU^{\mathcal{P}})^{i_2 i_1})\trans)}(\bbp_{i_1}, \bbp_{i_2}) \leq - \theta^2 \delta^{2N-2}$.
  If there exists a pair of player $i_1$ and $i_2$ such that this condition holds, then $C_{\bbU} \leq - \theta^2 \delta^{2N-2}$.

  \begin{align*}
    &C_{(\bbH(\bbU^{\mathcal{P}})^{i_1 i_2}, (\bbH(\bbU^{\mathcal{P}})^{i_2 i_1})\trans)}(\bbp_{i_1}, \bbp_{i_2})\\
    =& -\min_{g, h} \sum_{jl} x_{i_1 j}(\bbp_{i_1}) x_{i_2, l}(\bbp_{i_2}) (\bbH(\bbU^{\mathcal{P}})^{i_1 i_2} - g_j - h_k)^2 \\
    \leq& -\min_{g, h} \delta^2 \sum_{jl} (\bbH(\bbU^{\mathcal{P}})^{i_1 i_2} - g_j - h_k)^2 \\
    =& -\delta^2 \sum_{jl} (\bbH(\bbU^{\mathcal{P}})^{i_1 i_2} - g^*_j - h^*_k)^2 \\
    \leq& -\delta^2 (\theta \delta^{N-2})^2,
  \end{align*}
  where $\{g^*_j + h^*_k\}_{jk}$ is projection of  $\bbH(\bbU^{\mathcal{P}})^{i_1 i_2}$ on the trivial space.
  The first inequality follows as $\bbp \in S^{\delta}$; the second equality holds as the projection minimizing the distance to the trivial space, and the final inequality comes from the distance from $\bbH(\bbU^{\mathcal{P}})^{i_1 i_2}$ to the trivial space.
  
\end{itemize}

For all these cases, we can have OMWU is $C_{\bbU}$ to be strictly negative in domain $S^{\delta}$, which implies OMWU is Lyapunov chaotic in $S^{\delta}$.

\end{document}